\documentclass[12pt]{article}
\usepackage{jheppub}

\usepackage{amsmath,amssymb,amsthm}
\usepackage{mathtools}
\usepackage{graphicx,color}
\usepackage{mathrsfs}
\usepackage{cleveref}
\usepackage{simpler-wick}
\usepackage[shortcuts]{extdash}

\newtheorem{lemma}{Lemma}
\newtheorem{statement}{Statement}

\usepackage{ifpdf}

\DeclareMathOperator{\diag}{diag}
\newcommand{\dd}{\mathrm{d}}

\def\cO{\mathcal{O}}

\def\a{\alpha}

\newcommand{\ff}{\mathfrak{f}}
\newcommand{\fg}{\mathfrak{g}}
\newcommand{\bPsi}{\bar{\Psi}}
\newcommand{\bpsi}{\bar{\psi}}

\newcommand{\hh}{\widehat{h}}
\newcommand{\hht}{\widehat{\tilde{h}}}

\newcommand{\hL}{\widehat{L}}
\newcommand{\hLt}{\widehat{\tilde{L}}}

\DeclareMathOperator{\Tr}{Tr\!}
\DeclareMathOperator{\D}{d\!}

\let\Im\undefined

\DeclareMathOperator{\Im}{Im}
\newcommand{\ket}[1]{|#1\rangle}

\newcommand{\ketket}[2]{\ket{#1}\otimes\ket{#2}}
\newcommand{\braket}[2]{\langle#1|#2\rangle}
\newcommand{\braOket}[3]{\langle#1|#2|#3\rangle}
\newcommand{\expect}[1]{\langle #1 \rangle}

\newcommand{\be}{\begin{equation}}
\newcommand{\ee}{\end{equation}}
\newcommand{\ba}{\begin{eqnarray}}
\newcommand{\ea}{\end{eqnarray}}

\def\bpm{\begin{pmatrix}}
\def\epm{\end{pmatrix}}

\newcommand{\ov}{\overline}
\newcommand{\df}{\equiv}

\usepackage{relsize}

\newcommand{\rcolon}{\mathclose{:}}
\newcommand{\Lcolon}{\mathlarger{\mathopen{:}}}
\newcommand{\Rcolon}{\mathlarger{\mathclose{:}}}

\newcommand{\lcirc}{\mathopen{{}^{\circ}_{\circ}}}
\newcommand{\rcirc}{\mathclose{{}^{\circ}_{\circ}}}

\newcommand{\noeq}{\mathrel{\phantom{=}}}

\title{\centering 
Proving the Weak Gravity Conjecture \\ 
in Perturbative String Theory\\
{\Large --- Part I: The Bosonic String ---}}

\author[a]{Ben Heidenreich}
\author[b]{and Matteo Lotito}

\affiliation[a]{Amherst Center for Fundamental Interactions,\\Department of Physics, University of Massachusetts, Amherst, MA 01003 USA}
\affiliation[b]{Department of Physics and Astronomy \& Center for Theoretical Physics,
Seoul National University, Seoul 08826, Republic of Korea}
\affiliation[c]{Department of Physics, Korea Advanced Institute of Science and Technology, Daejeon 34141, Republic of Korea}

\emailAdd{bheidenreich@umass.edu}
\emailAdd{matteolotito@gmail.com}

\abstract{
We present a complete proof of the Weak Gravity Conjecture in any perturbative bosonic string theory in spacetime dimension $D\ge6$. Our proof works by relating the black hole extremality bound to long range forces, which are more easily calculated on the worldsheet, closing the gaps in partial arguments in the existing literature. We simultaneously establish a strict, sublattice form of the conjecture in the same class of theories. We close by discussing the scope and limitations of our analysis, along with possible extensions including an upcoming generalization of our work to the superstring.

}

\begin{document}

\makeatletter
\let\old@fpheader\@fpheader
\renewcommand{\@fpheader}{\old@fpheader\hfill
{ACFI-T24-01}}
\makeatother

\maketitle

\section{Introduction} \label{sec:intro}

The swampland program~\cite{Vafa:2005ui} seeks to distinguish effective field theories (EFTs) that arise as the low-energy limit of a consistent quantum gravity theory (QGT) from those that, while seemingly consistent at low energies, do not---categories known as the landscape and the swampland, respectively. Without a general understanding of quantum gravity, it is impossible to definitively place a given theory in the swampland. Nonetheless, 
  the low-energy effective field theories arising in numerous string theory examples share common features that have no obvious low-energy explanations. Motivated by this and by various suggestive arguments about black holes and other expected features of quantum gravity, a number of ``swampland conjectures'' have been formulated, see \cite{Palti:2019pca, vanBeest:2021lhn, Agmon:2022thq} for reviews.

One of the oldest and best-explored of these conjectures is the \textbf{Weak Gravity Conjecture} (WGC)~\cite{ArkaniHamed:2006dz} (see \cite{Palti:2020mwc, Harlow:2022gzl} for reviews). In its mildest form, requires that a QGT with a massless photon contains a charged particle with charge-to-mass ratio greater than or equal to that of any black hole of parametrically large mass and charge:\footnote{The WGC is often defined with reference to ``extremal'' black holes, but these do not always exist, and in particular they do not exist for the electrically charged black holes studied in this paper. Similar issues are discussed in \cite{Harlow:2022gzl}, appendix A.}
\begin{equation} \label{eqn:WGCsketch}
\text{WGC:}\qquad \frac{|q|}{m} \ge \frac{|Q|}{M} \biggr|_{\text{large BH}} \,.
\end{equation}
Here we assume zero cosmological constant.\footnote{See~\S\ref{sec:final} for a discussion of the WGC when a non-zero cosmological constant is generated by string loops.} 
By a black hole we mean an object with a macroscopic horizon---i.e., with curvature much less than the quantum gravity scale \cite{Dvali:2007hz,Dvali:2007wp}---that hides the source of the mass and charge. Note that, since it is defined using only large black holes, the WGC bound is not sensitive to (generally unknown) derivative corrections.

We call a particle satisfying~\eqref{eqn:WGCsketch} \emph{superextremal}. The existence of such a particle is essentially the same as the requirement that large non-extremal black holes are kinematically unstable and can decay to light particles. Thus, the WGC compares two fundamentally different things: (1) the exact spectrum of the QGT and (2) the space of large black hole solutions in the low-energy EFT, which is determined by the two-derivative effective action for the massless fields.






Note that the WGC bound~\eqref{eqn:WGCsketch} can be saturated, and in fact this is the only way to satisfy it in supersymmetric theories with BPS black holes, since then the BPS bound does not permit strictly superextremal states. In such cases, the charged particle satisfying the bound must also be BPS. Conversely, Ooguri and Vafa have conjectured that the WGC is \emph{never} saturated when the superextremal particle is not BPS~\cite{Ooguri:2016pdq}.\footnote{This does \emph{not} mean, however, that a BPS particle always saturates the WGC bound, as there may not be a corresponding BPS black hole, see, e.g.,~\cite{Alim:2021vhs}.} We refer to this version of the conjecture as the \textbf{Ooguri-Vafa WGC} for ease of reference.

The WGC has been verified in many string-derived QGTs, but in this paper we will go further. We aim to prove that all string theories satisfy the WGC. This is a very difficult problem, and we provide only a partial solution.
 We focus on  the tree-level string spectrum in perturbative string theory. The string modes carry electric charges under the NSNS sector gauge fields, and we will prove (finally completing an argument first given in \S4.1 of~\cite{ArkaniHamed:2006dz}) that they invariably satisfy the WGC for these gauge fields.

This proves the WGC in perturbative string theory, up to the following two caveats.
Firstly, string loop corrections might spoil the result. Specifically, if the WGC bound were \emph{saturated} at string tree-level, then string loop corrections to the mass of the (apparently) extremal charged particle could render it subextremal, depending on their sign. Besides calculating these corrections explicitly, this issue can be avoided by ensuring that either (a) the extremal particle is BPS, and thus its mass is not corrected or (b) the WGC bound is satisfied but not saturated at string tree-level, so that the string loop corrections are not large enough (at small string coupling) to make the particle subextremal. In other words, if the Ooguri-Vafa WGC holds at string tree-level, then perturbative corrections to the mass spectrum are tamed.

Secondly, there might be other one-form gauge fields outside the electric NSNS sector. Because string modes only carry electric NSNS sector gauge charge, we would need to consider D-branes, solitons or other objects to complete the charged spectrum in such theories, an interesting topic that is beyond the scope of our paper. Fortunately, in both heterotic and closed bosonic string theories, this does not occur in target spacetime dimension $D\ge 6$.\footnote{In $D=4,5$, there are magnetic NSNS one-form gauge fields (the magnetic dual of the Kalb-Ramond two-form $B_2$ in 5d, and the magnetic duals of all the one-form NSNS gauge fields in 4d).} Even when additional one-form gauge fields are present (such as in type II string theories), our results prove the WGC for an important subsector of the theory.

Note that our work is different from and complementary to attempts to prove the WGC using bottom-up reasoning (see~\cite{Harlow:2022gzl} for a review). This paper aims to fill gaps in the top-down evidence for the WGC; our work has nothing to say about hypothetical QGTs that are unrelated to string theory, but has strong implications for string theories.\footnote{These implications are not entirely limited to the \emph{perturbative} string for the usual reason that BPS states can sometimes be tracked from weak to strong coupling.}

\subsection{The WGC via repulsive forces}

To motivate our proof, it is convenient to discuss several other swampland conjectures that will also be proved as byproducts of our argument.
Closely related to the WGC is the \textbf{Repulsive Force Conjecture} (RFC)~\cite{ArkaniHamed:2006dz,Palti:2017elp,Heidenreich:2019zkl}, which requires the existence of a charged particle that repels its identical twin when far apart. In particular, the long-range self-force should not be attractive:
\begin{equation} \label{eqn:RFCsketch}
\text{RFC:}\qquad k q^2 - G_N m^2 - G_\phi \mu^2 \ge 0 \,,
\end{equation}
where $k$, $G_N$, and $G_\phi$ are the electromagnetic, gravitational, and scalar force constants, respectively, and $\mu = \frac{d m}{d\phi}$ is the scalar charge.  We call a particle satisfying~\eqref{eqn:RFCsketch} \emph{self-repulsive}. In four dimensional theories, the absence of such a particle implies the existence of
 infinite towers of stable Newtonian bound states~\cite{ArkaniHamed:2006dz}.\footnote{The same conclusion does not follow in $D>4$~\cite{Heidenreich:2019zkl}.}

The WGC and the RFC are related by the observation that an identical pair of extremal Reissner-Nordstr\"om black holes exert no long-range force on each other.
Thus, \eqref{eqn:RFCsketch} is the same as \eqref{eqn:WGCsketch} if we ignore scalar forces and assume that the extremal black hole has a Reissner-Nordstr\"om geometry. However, both assumptions are invalidated in theories with moduli: the moduli alter the black hole geometry, changing the WGC bound~\cite{Heidenreich:2015nta}, while simultaneously mediating long-range scalar forces, changing the RFC bound~\cite{Palti:2017elp,Heidenreich:2019zkl}. A priori, the scalar charge of a given massive particle has nothing to do with the moduli-induced modifications to the geometry of a black hole, hence the two conjectures become distinct once moduli are introduced~\cite{Heidenreich:2019zkl}.

Nonetheless, important links remain. Despite moduli-induced corrections to both the black hole geometry and to the long-range forces, extremal black holes still exert no long-range force on their identical copies~\cite{Heidenreich:2020upe}. More important for our paper is the following observation~\cite{Heidenreich:2019zkl,Harlow:2022gzl}, reviewed in~\S\ref{subsec:selfrepulsive}: if a particle exists that is self-repulsive throughout the entire moduli space, then that particle is superextremal. Thus, the existence of such a particle ensures that both the WGC and the RFC are satisfied.

The WGC holds in all string-derived QGTs where it has been checked to date. In fact, typically there is not just one superextremal charged particle but an infinite tower of them with increasing mass and charge~\cite{Heidenreich:2015nta,Heidenreich:2016aqi,Heidenreich:2017sim,Grimm:2018ohb,Andriolo:2018lvp,Lee:2018urn,Lee:2018spm}. This property, known as the \textbf{tower WGC}~\cite{Andriolo:2018lvp}, is related to the consistency of the WGC upon compactifying on a circle~\cite{Heidenreich:2015nta,Andriolo:2018lvp}. An even stronger statement seems to hold: the lattice\footnote{The charges form a lattice if the gauge group is compact (e.g., $U(1)^N$), where the compactness of the gauge group is yet another swampland conjecture~\cite{Banks:2010zn}.} of allowed charges $\Gamma_Q$ contains a sublattice $\Gamma_{\text{ext}} \subseteq \Gamma_Q$ of finite index that is completely populated by superextremal charged particles~\cite{Heidenreich:2016aqi}. This \textbf{sublattice WGC} holds in all string-derived examples (in dimension $D\ge 5$)\footnote{The conjecture requires modification in $D=4$~\cite{Heidenreich:2016aqi,Heidenreich:2017sim,Klaewer:2020lfg} (as does the tower WGC) due to the possibility of large logarithms in loop diagrams; see also~\cite{Lee:2019tst}. These subtleties do not occur for $D \ge 5$, making possible highly sensitive checks of the conjecture~\cite{Alim:2021vhs,Gendler:2022ztv}.} where it has so-far been checked.
To distinguish it from these stronger conjectures, the WGC, as we first defined it, is sometimes called the \textbf{mild WGC}.

When restricted to the tree-level string modes, the sublattice WGC is almost a consequence of worldsheet modular invariance~\cite{Heidenreich:2016aqi,Montero:2016tif} (see also~\cite{Lee:2018spm,Aalsma:2019ryi}). As reviewed in~\S\ref{subsec:specflow}, one can show that there is a finite-index sublattice of the NSNS sector electric charge lattice that is completely populated by charged string modes satisfying
\begin{equation} \label{eqn:spectralflowbound}
\frac{\alpha'}{4} m^2 \le \frac{1}{2} \max(Q_L^2, Q_R^2) \,.
\end{equation}
In particular, a finite-index sublattice of states saturating this bound can be obtained by spectral flow from the graviton. However, this is not quite a proof of the sublattice WGC.
The crucial missing link, supplied in this paper, is the connection between this bound on the charge-to-mass ratio and the black hole extremality bound.

To make this connection, suppose that we instead try to prove the \textbf{sublattice RFC}~\cite{Heidenreich:2019zkl}, i.e., the analogous conjecture with ``superextremal'' replaced by ``self\-/repulsive''. To do so, it would be sufficient to show that the states generated by spectral flow from the graviton are self-repulsive. In fact, this would have larger consequences, as these self-repulsive states necessarily persist throughout the moduli space, implying (as reviewed in~\S\ref{subsec:selfrepulsive}) that they are superextremal~\cite{Heidenreich:2019zkl,Harlow:2022gzl}. Thus, proving that the states related to the graviton by spectral flow are self-repulsive would prove both the sublattice RFC and the sublattice WGC, implying, respectively, the (mild) RFC and the (mild) WGC.

This may seem like a roundabout approach to proving the mild WGC, which is a much weaker statement than the sublattice conjectures discussed above. In fact, the sublattice conjectures are in an important sense much easier to handle than the mild WGC or RFC in QGTs with more than one massless photon (which occur very frequently in the landscape). To state the mild WGC precisely in such examples, it is convenient to express it in terms of the charge-to-mass vectors $\vec{z} = \vec{q}/m$ of various states. In $\vec{z}$-space, parametrically large black holes occupy a finite ``black hole region,'' and we call states outside or on the boundary of this region superextremal. The mild WGC now requires not just one superextremal charged particle, but rather enough of them to enclose the entire black hole region in their convex hull~\cite{Cheung:2014vva}, see figure~\ref{fig:convexhull}.\footnote{More precisely, along every rational direction in charge space, there should be a superextremal multiparticle state, where a ``multiparticle state'' is just a formal combination of particles with mass and charge equal to the total mass and charge of the particles~\cite{Heidenreich:2019zkl}.}

\begin{figure}\centering
\includegraphics[width=0.5\textwidth]{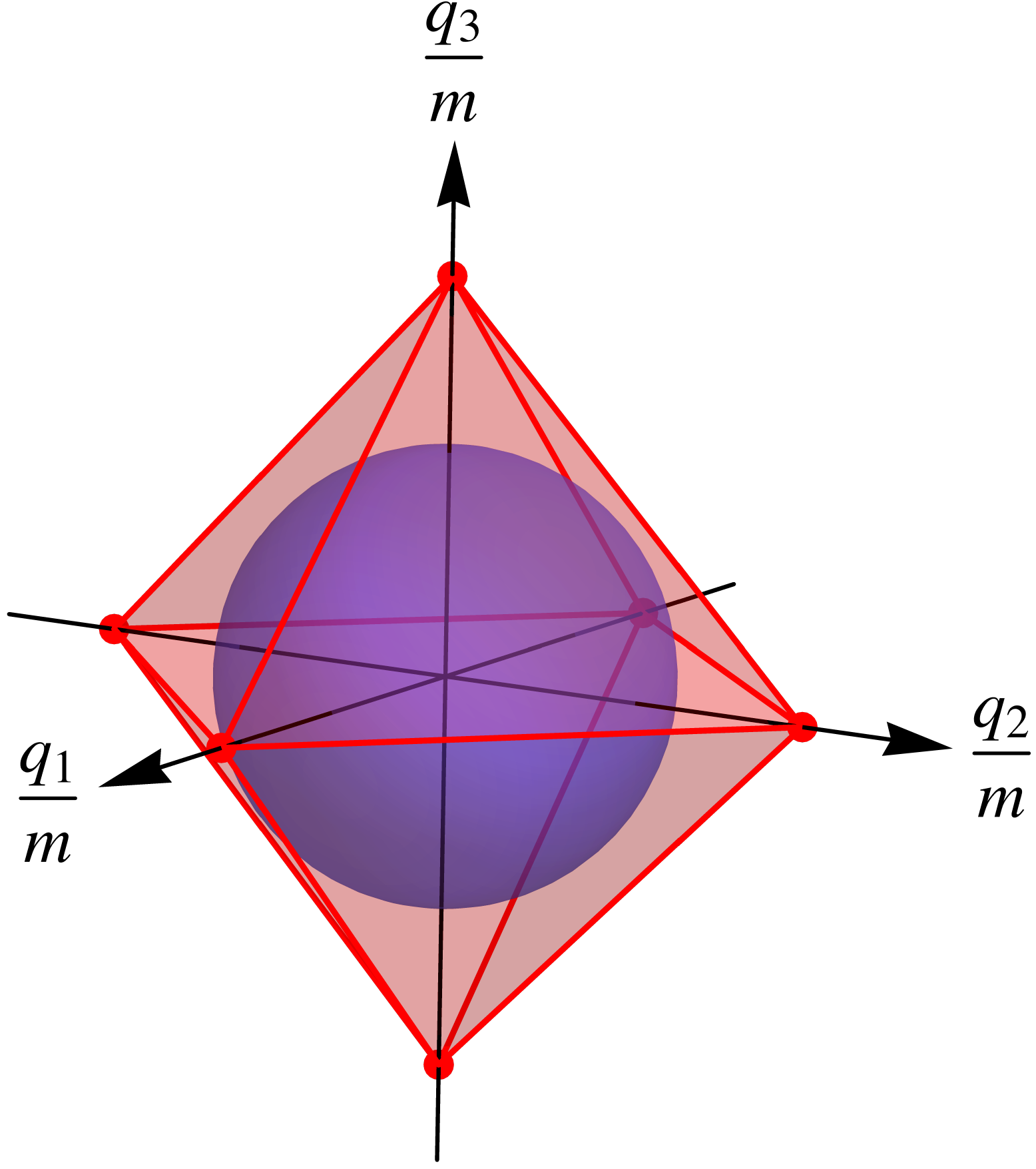}
\caption{The mild WGC requires the convex hull of the charge-to-mass ratios of all charged particles (red) to enclose the region spanned by charged black hole solutions (blue). This convex hull condition is marginally satisfied in the pictured example, as the center of each triangular face of the convex hull touches the outer surface of the ball-shaped black hole region. More generally, checking the convex hull condition in a theory with many gauge fields is a very non-trivial problem in high-dimensional geometry. Deriving the mild WGC from its stronger, sublattice variant---as in the present work---completely sidesteps this issue.}\label{fig:convexhull}
\end{figure}

This \textbf{convex hull condition} is simple and intuitive, but it involves high\-/dimensional geometry when there are more than a few massless photons, and consequently it can be quite hard to verify in such examples. It is much easier to check whether a given particle is superextremal. To do so, we need only compare its charge-to-mass ratio to that of a parametrically large black hole with a parallel charge vector. Thus, it is much easier to verify the sublattice WGC when the full spectrum is known, implying the mild WGC as a consequence without needing to consider any details of the $\vec{z}$-space geometry. 

\bigskip

With this reasoning in mind, the main new technical result required for our proof is to establish that the string modes related to the graviton by spectral flow are self-repulsive. In fact, we will show that those that saturate~\eqref{eqn:spectralflowbound} have vanishing self-force (a special case of ``self-repulsive''). In the bosonic case, we find that there is also a slightly lighter tower, with mass
\begin{equation}
\frac{\alpha'}{4} m^2_{\text{lightest}} = \frac{1}{2} \max(Q_L^2, Q_R^2) - 1 \,.
\end{equation}
These lighter states may or may not be self-repulsive (depending on $Q_L^2 - Q_R^2$), but our analysis implies that they are \emph{strictly} superextremal. Thus, the tree-level spectrum satisfies the \textbf{strict, sublattice WGC}: the charge lattice $\Gamma_Q$ contains a finite-index sublattice $\Gamma_{\text{ext}} \subseteq \Gamma_Q$ completely populated by particles that are either BPS or strictly superextremal. This implies and strengthens the Ooguri-Vafa WGC, ensuring that perturbative corrections do not spoil the tree-level result.\footnote{We defer further discussion of the Ooguri-Vafa WGC in superstring theories to~\cite{BMSuperstring}.}

\bigskip

For simplicity and clarity we have separated our proof into two parts, addressing the bosonic string in the present paper and the superstring in the sequel~\cite{BMSuperstring}. For completeness, we will review many basic facts about 2d CFTs, as well as providing a new, more explicit derivation of the universal modular transformation of the flavored partition function. While these facts are, in some sense, elementary, they are critical to establishing our novel results on the universal behavior of long-range forces and extremal black holes in perturbative (bosonic) string theories. 

The outline of this paper is as follows. In~\S\ref{sec:2} we summarize our proof strategy and discuss the computation of long-range forces on the worldsheet. In~\S\ref{sec:3}, we work through the entire proof for the bosonic string. Along the way, we review and make more rigorous some existing arguments on which our proof relies. 
In~\S\ref{sec:final} we summarize our final results, discuss their limitations and possible extensions, and provide a brief preview of the superstring generalization, to appear in~\cite{BMSuperstring}.
Several technical aspects of the proof are included in appendices. In appendix~\ref{app:NormalOrder} we describe the procedure for normal-ordering holomorphic operators on the worldsheet, in appendix~\ref{app:sugawara} we review the Sugawara construction, and in appendix~\ref{app:toroidal} we derive the long range forces between string modes in a toroidal compactification of the bosonic string via the low-energy effective action.

\section{Proof strategy}
\label{sec:2}
We first summarize our proof strategy, as developed above. Our argument relies on three key facts, established (for the bosonic string) in the indicated sections: 
\begin{enumerate}
\item
There exists an infinite tower of charged particles related to the graviton by spectral flow, which saturate~\eqref{eqn:spectralflowbound} and fill out a finite-index sublattice of the electric NSNS charge lattice (\S\ref{subsec:specflow}).
\item
These particles are self-repulsive (\S\ref{subsec:bosonic-univ}). 
\item
A particle that is self-repulsive everywhere in moduli space is superextremal (\S\ref{subsec:selfrepulsive}).
\end{enumerate}
Taken together, these facts imply both the sublattice RFC and the sublattice WGC, which in turn imply the mild RFC and WGC, respectively. While our proof is at string tree-level, it is relatively robust against perturbative corrections, as discussed in~\S\ref{sec:final}.

While all three facts listed above are essential to our proof, the second is the most novel (the others having appeared in some form in previous literature). As such, we begin by explaining in general terms an efficiently strategy for computing long-range forces on the worldsheet, saving the (bosonic string) implementation for~\S\ref{subsec:bosonic-univ}.

\subsection{Computing long-range forces on the worldsheet} \label{subsec:worldsheetforce}

There are two natural approaches to computing long-range forces, one based on the low-energy effective action and the other on the S-matrix.

First, consider the low-energy effective action. Fermions, massive fields, and charged fields do not mediate long-range forces between identical particles. Truncating to the massless, neutral bosons, the two-derivative effective action takes the general form:
\begin{equation} \label{eqn:gen-action-EFT}
S = \int \!\dd^d x\, \sqrt{-g} \biggl[\frac{1}{2 \kappa_d^2} R -\frac{1}{2} G_{i j}(\phi) \nabla\phi^i \cdot\nabla\phi^j - \frac{1}{2} \ff_{a b}(\phi) F_2^a \cdot F_2^b - V(\phi) \biggr] + \ldots \,,
\end{equation}
where $V(\phi)$ is a scalar potential of quartic or higher order and the omitted terms involve $p$-form gauge fields, theta terms, and Chern-Simons terms, none of which contribute to the long-range forces between electrically charged particles (see, e.g., \cite{Heidenreich:2020upe} for a careful accounting).

Note that the scalars $\phi^i$ are massless by assumption, but they are not necessarily \emph{moduli}. The moduli are the flat directions of $V(\phi)$ (parameterizing the vacuum manifold $V(\phi) = 0$).\footnote{Here we assume for simplicity that the cosmological constant vanishes (i.e., $V=0$ in the vacuum).} Massless scalars contribute to long-range forces whether they are moduli or not. Moduli, however, play a special role in black hole solutions, as discussed in~\S\ref{subsec:selfrepulsive}.

In terms of the couplings appearing in~\eqref{eqn:gen-action-EFT}, the long-range force between a pair of particles of mass $m(\phi)$, $m(\phi)'$ and electric charge $q_a$, $q_a'$ is
\begin{equation}
\vec{F} = \frac{\mathcal{F}}{V_{d-2} r^{d-2}} \hat{r} \qquad \text{where} \qquad
\mathcal{F} \equiv \mathfrak{f}^{a b} q_a q_b' - k_N m m' - G^{i j}  \frac{\partial m}{\partial
  \phi^i}  \frac{\partial m'}{\partial \phi^j} , \label{eqn:selfforce}
\end{equation}
where $V_{d-2}$ is the volume of the unit $(d-2)$-sphere, $k_N = \frac{d - 3}{d - 2} \kappa_d^2$ is the rationalized Newton's
constant and $\mathfrak{f}^{ab}$, $G^{ij}$ are the inverses of $\mathfrak{f}_{ab}$, $G_{ij}$, respectively. Thus, the long-range force is determined by the couplings $\mathfrak{f}_{a b}$, $k_N$ and $G_{i j}$ as well as the mass $m$, charge $q_a$, and scalar charge $\frac{\partial m}{\partial \phi^i}$ of the particles in question.

We use this effective-action description of the long-range forces to establish the connection between super-extremality and self-repulsiveness in~\S\ref{subsec:selfrepulsive}. However, to establish whether a given string mode is self-repulsive, an equivalent S-matrix formulation is more expedient, as follows.

Elastic scattering between two particles is described by a four-point amplitude, but the long-range component of the interaction is distinguished by intermediate massless force-carrying bosons going on shell, factoring the four-point amplitude into two three-point amplitudes. Thus, to calculate the long-range forces it is sufficient to calculate the three point amplitudes for each particle to emit/absorb each force-carrier, see Figure~\ref{fig:Feyn1}.

\begin{figure}
\includegraphics[width=\textwidth]{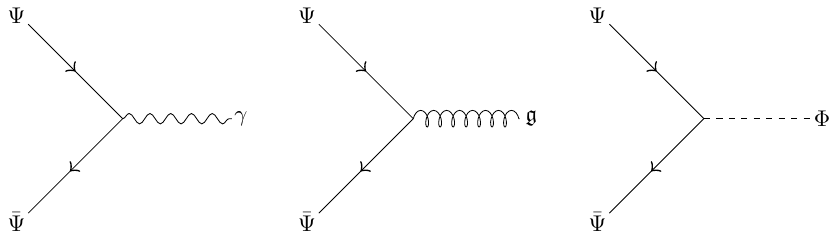}
\caption{The three-point amplitudes used to calculate the long-range forces. \label{fig:Feyn1}}
\end{figure}
To make this explicit, we first canonically normalize the EFT fields
\begin{equation}
\phi^i = \phi^i_0 + g^i_I \Phi^I\,, \qquad
A^a_{\mu} = \mathfrak{e}^a_A \gamma^A_{\mu} \,, \qquad g_{\mu \nu} = \eta_{\mu \nu} + \kappa_d \mathfrak{g}_{\mu \nu}\,,
\end{equation}
where $\phi^i_0$ denotes the chosen vacuum and $g_i^I$, $\mathfrak{e}_a^A$ are vielbeins for $G_{i j}$ and $\mathfrak{f}_{a b}$ in that vacuum (so that $G_{i j}(\phi_0) = g_i^I g_j^J \delta_{I J}$, $\mathfrak{f}_{a b}(\phi_0) = \mathfrak{e}_a^A \mathfrak{e}_b^B \delta_{A B}$), with inverse vielbeins  $g^i_I$, $\mathfrak{e}^a_A$. The three-point amplitudes between the particle $\Psi$ in question and the force carrying bosons $\gamma^A$, $\mathfrak{g}_{\mu \nu}$ and $\Phi^I$ (see Figure~\ref{fig:Feyn1}) then fix $\mathfrak{e}^a_A q_a$, $\kappa_d m$, and $g_I^i \frac{\partial m}{\partial \phi^i}$, respectively, from which the long-range force~\eqref{eqn:selfforce} can be calculated.

These three-point amplitudes correspond to the worldsheet three-point functions $\expect{\bPsi \Psi \gamma^A}$, $\expect{\bPsi \Psi \fg}$ and $\expect{\bPsi \Psi \Phi^I}$, respectively, where in an abuse of notation we denote vertex operators by the same symbols as their corresponding canonically-normalized EFT fields.\footnote{To be precise, the string scattering amplitude involves a conformal gauge-fixing determinant, i.e., conformal ghosts in a BRST approach. We suppress these (universal) details for simplicity, but note that when they are properly accounted for the closed string three-point function does not depend on the positions of the operators. Thus, we need not specify these positions.} Thus, in a canonically normalized basis
\begin{align}
\expect{\bPsi \Psi \gamma^A} &\sim \delta^{A B} \mathfrak{e}^a_B q_a\,, & \expect{\bPsi \Psi \fg} &\sim \kappa_d m\,, & \expect{\bPsi \Psi \Phi^I} &\sim \delta^{I J} g_J^i \frac{\partial m}{\partial \phi^i} \,, \label{eqn:NDA3points}
\end{align}
for non-relativistic $\Psi$, up to momentum and/or polarization-dependent factors.
The long-range force is then
\begin{equation}
\mathcal{F} \sim \expect{\bPsi \Psi \gamma^A} \delta_{A B} \expect{\gamma^B \bPsi' \Psi'} - \expect{\bPsi \Psi \fg}\expect{\fg \bPsi' \Psi'} - \expect{\bPsi \Psi \Phi^I} \delta_{I J} \expect{\Phi^J \bPsi' \Psi'} \,. \label{eqn:naiveforceCFT}
\end{equation}
However,~\eqref{eqn:naiveforceCFT} only applies if we correctly normalize both the vertex operators as well as an overall functional-determinant factor (see, e.g., \cite{Polchinski:1998rq}).
This is a non-trivial process involving, e.g., unitarity cuts, with the answer depending on the string coupling.\footnote{In fact, these normalizing factors are the \emph{only} place in which the string coupling shows up in the tree-level string S-matrix.}

\begin{figure}
\includegraphics[width=\textwidth]{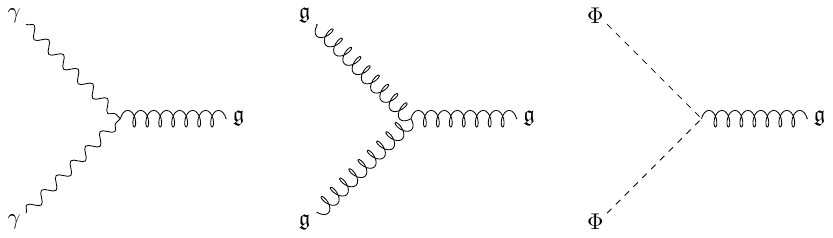}
\caption{The three-point amplitudes used to correctly normalize the vertex operators. \label{fig:Feyn2}}
\end{figure}
To avoid this complication, consider the amplitude for each force-carrier to emit or absorb a graviton, see Figure~\ref{fig:Feyn2}. Observe that:
\begin{align}
\expect{ \gamma^A \gamma^B \fg } &\sim \kappa_d \delta^{AB} \,, & \expect{ \fg\fg\fg } &\sim \kappa_d \,, & \expect{ \Phi^I \Phi^J \fg }&\sim \kappa_d \delta^{I J}\,,
\end{align}
in a canonically normalized basis, up momentum/polarization dependent factors as before. Therefore,
\begin{equation}
\frac{\mathcal{F}}{\kappa_d} \sim
\expect{\bPsi \Psi \gamma^A} \expect{\gamma \gamma \fg}^{-1}_{A B} \expect{\gamma^B \bPsi' \Psi'}
- \frac{\expect{\bPsi \Psi \fg}\expect{\fg \bPsi' \Psi'}}{\expect{\fg\fg\fg}}
- \expect{\bPsi \Psi \Phi^I} \expect{\Phi \Phi \fg}^{-1}_{I J} \expect{\Phi^J \bPsi' \Psi'} \,, \label{eqn:selfforceCFT}
\end{equation}
where $\expect{\gamma \gamma \fg}^{-1}_{A B}$ and $\expect{\Phi \Phi \fg}^{-1}_{I J}$ denotes the inverses of the matrices $\expect{\gamma^A \gamma^B \fg}$ and $\expect{\Phi^I \Phi^J \fg}$, respectively. The right-hand side of~\eqref{eqn:selfforceCFT} is now independent of the choice of basis for the gauge boson and scalar vertex operators $\gamma^A$ and $\Phi^I$, respectively. Moreover, it scales homogeneously with the normalizations of the other vertex operators and with the overall functional determinant factor (which affects every three-point function equally). Thus, even if we fail to correctly normalize the vertex operators, \eqref{eqn:selfforceCFT} still gives the right long-range force up to an overall factor.

The overall normalization can be fixed by normalizing the gravitational term to $-k_N m m'$ as in~\eqref{eqn:selfforce}.
The final result is
\begin{equation}
\mathcal{F} = k_N m m' \biggl[ \frac{\expect{\bPsi \Psi \gamma^A} \expect{\gamma \gamma \fg}^{-1}_{A B} \expect{\gamma^B \bPsi' \Psi'}}{\expect{\bPsi \Psi \fg} \expect{\fg\fg\fg}^{-1} \expect{\fg \bPsi' \Psi'}}
- 1
- \frac{\expect{\bPsi \Psi \Phi^I} \expect{\Phi \Phi \fg}^{-1}_{I J} \expect{\Phi^J \bPsi' \Psi'}}{\expect{\bPsi \Psi \fg} \expect{\fg\fg\fg}^{-1} \expect{\fg \bPsi' \Psi'}} \biggr] \,, \label{eqn:selfforceCFT2}
\end{equation}
up to momentum/polarization-dependent factors in the first and last terms that we have not properly accounted for. The result~\eqref{eqn:selfforceCFT2} is independent of the normalizations of all the vertex operators and of the overall functional determinant factor. Thus, we need not compute these normalizations; it is sufficient to know the mass of $\Psi$ as well as the three-point functions $\expect{ \bPsi \Psi \gamma^A }$, $\expect{ \bPsi \Psi \fg }$, $\expect{ \bPsi \Psi \Phi^I }$, $\expect{ \gamma^A \gamma^B \fg }$, $\expect{ \fg\fg\fg }$, and $\expect{ \Phi^I \Phi^J \fg }$ in any convenient basis.

However, as written~\eqref{eqn:selfforceCFT2} still contains unspecified momentum/polarization dependent factors in the first and third terms that need to be correctly handled to extract the long-range force. To avoid dealing with these factors explicitly, first note that the worldsheet CFT has two components, an ``external'' CFT describing the target space of the $d$-dimensional theory and an ``internal'' CFT describing, e.g., a compactification of the critical string. The external CFT is \emph{universal}, in that it depends only on the spacetime dimension and the class of string theory we are considering (bosonic, heterotic, type II, etc.), whereas the internal CFT encodes all of the details of the theory (e.g., the choice of compactification manifold). For example, for a bosonic string theory in a $d$-dimensional Minkowski background, the external CFT is simply $d$ free bosons, one of which is timelike.

The vertex operators and their $n$-point functions likewise split into external and internal CFT components, where the external components encode the momentum, polarization and other universal properties. Thus, rather than computing the external CFT three-point functions, we fold these factors into the unspecified coefficients of the first and third terms of~\eqref{eqn:selfforceCFT2}, keeping only the internal CFT portions of each expression. Since the unspecified coefficients are universal---i.e., they depend only on external CFT data---we can then fix them by comparing with known results in a representative string theory with the same external CFT.

Thus, to compute the long-range forces on $\Psi$, we need only know its mass as well as the \emph{internal} portions of the three-point functions $\expect{ \bPsi \Psi \gamma^A }$, $\expect{ \bPsi \Psi \fg }$, $\expect{ \bPsi \Psi \Phi^I }$, $\expect{ \gamma^A \gamma^B \fg }$, $\expect{ \fg\fg\fg }$, and $\expect{ \Phi^I \Phi^J \fg }$ in any convenient basis. The rest will be fixed by comparing with a specific string theory of the same type where the long-range forces are known, such as bosonic string theory compactified on a torus, see appendix~\ref{app:toroidal}.

With this general strategy in mind, we now lay out our proof in detail for the bosonic string, leaving the superstring case to forthcoming work~\cite{BMSuperstring}.

\section{Proof for the bosonic string}
\label{sec:3}
In this section, we prove the WGC in bosonic string theory, following the strategy outlined above. We first describe how the gauge bosons, graviton, massless scalars and massive charged particles appearing in the low-energy EFT match with the bosonic worldsheet description (\S\ref{subsec:worldsheetEFT}). We then review and sharpen the arguments of~\cite{Heidenreich:2016aqi,Montero:2016tif}, showing that a tower of massive charged particles with $\alpha' m^2 \sim Q^2$ exists (\S\ref{subsec:specflow}). Next, we apply the strategy laid out above to determine the long-range forces between these charged particles (\S\ref{subsec:bosonic-univ}). Finally, we connect self-repulsiveness and superextremality to show that these particles are superextremal (\S\ref{subsec:selfrepulsive}).

\subsection{Worldsheet spectrum and interactions} \label{subsec:worldsheetEFT}

We begin by reviewing some details of the bosonic worldsheet theory for ease of reference and to establish conventions.

The worldsheet of an arbitrary bosonic string theory in $D$ flat target space dimensions is described by a compact, unitary, modular-invariant ``internal'' CFT of central charge $c=\tilde{c} = 26-D$ together with $D$ ``external'' free bosons (one timelike) and the conformal ghosts. For concreteness, we use old covariant quantization, where the ghosts are ignored and physical states are required to be Virasoro primaries of weight $(h,\tilde{h})~=~(1,1)$, with primaries differing by a null state declared to be equivalent.\footnote{Null states arise because the external CFT is non-unitary due to the timelike free boson.}
These states can be decomposed into ``external'' and ``internal'' tensor factors:
\begin{equation}
\ket{\Psi} = \sum_i \ket{\Psi_i}_{\text{ext}} \otimes  \ket{\Psi_i}_{\text{int}} \,,
\end{equation}
where $\ket{\Psi_i}_{\text{ext}}$ and $\ket{\Psi_i}_{\text{int}}$ are (not necessarily primary) states of the external and internal CFT, respectively, with weights adding to $(1,1)$. External states are constructed using raising operators $\alpha_{-m}^\mu$ and $\tilde{\alpha}_{-n}^\nu$ ($m,n>0$, of weight $(m,0)$ and $(0,n)$, respectively) applied to the target-space momentum-$k^\mu$ ground state $\ket{0; k}$ (itself of weight $h = \tilde{h} = \frac{\alpha'}{4} k^2$). 

Massless bosons mediate long-range forces in the low energy effective theory. Since the internal-CFT states have non-negative weights by unitarity, these arise in four possible ways: 
\begin{equation}\label{eq:worldhseetstates}
\begin{alignedat}{5}
\mathrm{I})&\quad&& \a_{-1}^\mu  \tilde{\a}_{-1}^\nu \ketket{0; k}{1} \,, &\qquad&&  \mathrm{II}) &\quad&& \a_{-1}^\mu  \ketket{0; k}{\tilde{J}}\,,  \\
\mathrm{III})&&& \tilde{\a}_{-1}^\nu \ketket{0; k}{J}\,, &&& \mathrm{IV})&&&  \ketket{0; k}{\varphi} \,,
\end{alignedat}
\end{equation}
where $k^2=0$ and $1$, $J(z)$, $\tilde{J}(\bar{z})$, and $\varphi(z,\bar{z})$ are internal-CFT operators of weight $(0,0)$, $(1,0)$, $(0,1)$, and $(1,1)$, respectively, all Virasoro primaries,\footnote{All states with $h,\tilde{h} \le 1$ are Virasoro primaries, since by unitarity $L_n \ket{\psi} = 0$ for $n > h_{\psi}$, whereas for $h_{\psi}=1$, $\lVert L_1 \ket{\psi} \rVert^2 = \braOket{\psi}{L_{-1}L_1}{\psi} = 0 \Longrightarrow L_1\ket{\psi}=0$ using the fact that $\ket{\chi} \equiv L_1 \ket{\psi}$ has $h_\chi = 0$, hence $\lVert L_{-1} \ket{\chi} \rVert^2 = \braOket{\chi}{L_1 L_{-1}}{\chi} = \braOket{\chi}{[L_1, L_{-1}]}{\chi} = 2 h_{\chi} \braket{\chi}{\chi} = 0 \Longrightarrow L_{-1}\ket{\chi} = L_{-1} L_1 \ket{\psi}= 0$.}
 and $\ket{1}$, $\ket{J}$, $\ket{\tilde{J}}$, and $\ket{\varphi}$ are the corresponding states. Any compact unitary CFT has a unique weight-$(0,0)$ operator (the identity operator), hence (I) has a unique instance in every theory, giving the target-space graviton, dilaton, and (for oriented strings) NSNS two-form $B_{\mu \nu}$. Likewise, weight-$(1,0)$ and $(0,1)$ operators are left and right-moving conserved currents, hence (II) and (III) give a target-space gauge boson $A_\mu^a$ for each continuous global symmetry of the internal CFT. Finally, (IV) gives a massless target-space scalar for each marginal ($h=\tilde{h}=1$) operator in the internal CFT.\footnote{The scalar is a modulus if and only if $\varphi$ is exactly marginal. However, all massless scalars (not just moduli) contribute to long-range forces, so we will \emph{not} assume exact marginality.}

Similarly, for each internal-CFT primary $\ket{\psi}$ charged under the worldsheet global symmetries, there is a family of target-space charged particles involving $\ket{\psi}$ and its descendants. 
The simplest of these involve no $\ket{\psi}$-descendants, and take the form
\begin{equation}\label{eqn:GeneralChargeState}
\ket{\Psi} = (\ldots) \ketket{0;k}{\psi} \,,
\end{equation}
where $(\ldots)$ denotes some combination of external-CFT raising operators $\alpha_{-m}^\mu$ and $\tilde{\alpha}_{-n}^\nu$ of weight $(h_\alpha,\tilde{h}_\alpha)$ satisfying $h_{\alpha} - \tilde{h}_{\alpha} = -h_{\psi} + \tilde{h}_\psi$ (where $h_\psi - \tilde{h}_\psi \in \mathbb{Z}$ follows from modular invariance). Since $\ket{\Psi}$ must have weight $(1,1)$ and $h_\alpha, \tilde{h}_\alpha$ are non-negative integers, we read off the spectrum
\begin{equation}
\frac{\alpha'}{4} m^2 = \max(h_\psi, \tilde{h}_\psi) +N-1 \,, \label{eqn:psimass}
\end{equation}
where $N$ is any non-negative integer.\footnote{We assume $D \ge 3$, ensuring that the Virasoro constraints can be solved for any $N \ge 0$ (this can be seen, e.g., in light-cone gauge).} The physical states involving descendants of $\ket{\psi}$ are more complicated to write, but fit within the same spectrum.

Each state in the spectrum has a CPT conjugate with vertex operator $V_{\text{CPT}}(k) = \ov{V}(-k)$, where $\ov{\mathcal{O}}$ denotes the \emph{conjugate} of a local operator $\mathcal{O}$. To define this, first consider the \emph{Euclidean adjoint}, which is the Wick rotation of the ordinary adjoint. The Euclidean adjoint is an anti-linear map that incorporates reversal of Euclidean time, so that $w^\dag = w$ in cylindrical quantization (where $w = \sigma + i t_E$) whereas $z^\dag = 1/z$ in radial quantization (where $z=e^{-i w}$).\footnote{Note that the Euclidean adjoint of a holomorphic operator remains holomorphic despite the anti-linearity, consistent with the fact that the adjoint of a left-mover is a left-mover on the Lorentzian worldsheet.}

A Lorentzian self-adjoint operator $\mathcal{O}^\dag(x) = \mathcal{O}(x)$ Wick rotates to a cylindrically quantized Euclidean self-adjoint operator $[\mathcal{O}(w,\bar{w})]^\dag = \mathcal{O}(w,\bar{w})$. More generally, we define the conjugate local operator $\ov{\mathcal{O}}(w,\bar{w})$ to equal $[\mathcal{O}(w,\bar{w})]^\dag$ \emph{in cylindrical quantization}. Then if $\mathcal{O}$ is a weight $(h,\tilde{h})$ primary operator, we obtain
\begin{equation}
[\mathcal{O}(z,\bar{z})]^\dag = (-z^2)^{h} (-\bar{z}^2)^{\tilde{h}} \bar{\mathcal{O}}(z,\bar{z}) \,, \label{eqn:selfadjointradial}
\end{equation}
in radial quantization, since $z^\dag = 1/z$. Thus, in radial quantization a Lorentzian self-adjoint operator $\mathcal{O}$ is self-conjugate, $\ov{\mathcal{O}}(z,\bar{z}) = \mathcal{O}(z,\bar{z})$ but not Euclidean self-adjoint.

Just as moving an operator $\mathcal{O}$ to the origin in radial quantization creates a ket (initial state) $|\mathcal{O}\rangle$ by the state-operator correspondence, likewise taking $\mathcal{O}^\dag$ to infinity on the Riemann sphere creates the corresponding bra (final state) $\langle\mathcal{O}|$, i.e.,
\begin{equation}
\lim_{z\to0} \mathcal{O}(z,\bar{z}) \longrightarrow |\mathcal{O}\rangle \qquad \text{is equivalent to} \qquad \lim_{z\to\infty} [\mathcal{O}(z,\bar{z})]^\dag \longrightarrow \langle \mathcal{O}| \,. \label{eqn:braradial}
\end{equation}
Thus, radially-quantized matrix elements with $n-2$ local operators inserted are related to $n$-point functions on the complex plane via\footnote{When applying this expression to a string amplitude, $V_1$ and $V_n$ must be fixed (rather than integrated) vertex operators. Then $h_1 = \tilde{h}_1 = 0$ in BRST quantization due to the $bc$ ghost contribution, and the $z$-dependent prefactor drops out. Equivalently, in old covariant quantization $h_1 = \tilde{h}_1 = 1$ and the $|z_1|^4$ prefactor is part of the gauge-fixing determinant.}
\begin{equation}
\langle \bar{V}_1 | V_2(z_2,\bar{z}_2) \cdots |V_n\rangle = \lim_{z_1 \to \infty} (-z_1^2)^{h_1} (-\bar{z}_1^2)^{\tilde{h}_1} \langle V_1(z_1,\bar{z}_1) V_2(z_2,\bar{z}_2) \cdots  V_n(0) \rangle \,. \label{eqn:matrixElementCorrelator}
\end{equation}
For neutral force carriers, the internal vertex operators $1$, $J$, $\tilde{J}$, $\varphi$ can be chosen to be self-conjugate, hence accounting for external factors their CPT conjugates differ only by their helicities. For charged states, the internal vertex operator $\psi$ has a conjugate $\bar{\psi}$ of the opposite charge.

\subsection{Modular invariance and spectral flow} \label{subsec:specflow}

We now show that the spectrum of a bosonic string theory with a conserved current $J(z)$ contains an infinite tower of charged particles with masses $\alpha' m^2 \sim Q^2$. A heuristic argument, first articulated in~\cite{ArkaniHamed:2006dz}, is as follows: the conserved current satisfies $\ov\partial J = 0$ as an operator equation. This suggests that we can write $J(z) = \partial \phi$ for a compact free boson $\phi(z) \sim \phi(z) + 2\pi f$. An infinite tower of charged operators $e^{i n \phi(z)/f}$, $n\in\mathbb{Z}$ can then be constructed.

This can be made precise using modular invariance, as shown in~\cite{Heidenreich:2016aqi,Montero:2016tif}. The partition function of the internal CFT on the torus $w\cong w+2\pi \cong w + 2\pi \tau$ takes the form
\begin{equation}
Z(\tau, \bar{\tau}) = \Tr\Bigl[ q^{L_0 - \frac{c}{24}} \bar{q}^{\tilde{L}_0 - \frac{\tilde{c}}{24}}\Bigr]  = \sum q^{h-\frac{c}{24}} \bar{q}^{\tilde{h}-\frac{\tilde{c}}{24}} \,, \qquad q \df e^{2\pi i \tau} \,, \label{eqn:Ztau}
\end{equation}
where the sum includes both primaries and descendants. Modular invariance requires
\begin{equation}
Z(\tau, \bar{\tau}) = Z(\tau+1, \bar{\tau}+1) = Z\Bigl(-\frac{1}{\tau}, -\frac{1}{\bar{\tau}}\Bigr) \,. \label{eqn:ZtauMT}
\end{equation}
(This implies, for instance, that $h - \tilde{h} \in \mathbb{Z}$ for every state in the spectrum.)
Likewise, the torus one-point function of a weight-$(h,\tilde{h})$ primary transforms as
\begin{equation}
Z[\mathcal{O}(w,\bar{w})](\tau,\bar{\tau}) = Z[\mathcal{O}(w,\bar{w})](\tau+1,\bar{\tau}+1) = \frac{1}{\tau^h \bar{\tau}^{\tilde{h}}} Z[\mathcal{O}(w/\tau,\bar{w}/\bar{\tau})]\Bigl(-\frac{1}{\tau},-\frac{1}{\bar\tau}\Bigr) \,. \label{eqn:ZOtauMT}
\end{equation}
More generally, for $n$ primary operators\footnote{It is sufficient for $\mathcal{O}_i$ to be a \emph{conformal} (not necessarily Virasoro) primary, since the M\"obius transformation $w' = w/\tau$ lies within the conformal group $SL(2,\mathbb{C})$.} $\mathcal{O}_i$ of weights $(h_i, \tilde{h}_i)$ inserted at positions $w_i$,
\begin{multline}
Z[\mathcal{O}_1(w_1,\bar{w}_1) \cdots \mathcal{O}_n(w_n,\bar{w}_n)](\tau,\bar{\tau}) = Z[\mathcal{O}_1(w_1,\bar{w}_1) \cdots \mathcal{O}_n(w_n,\bar{w}_n)](\tau+1,\bar{\tau}+1) \\
= \frac{1}{\tau^{\sum_i h_i} \bar{\tau}^{\sum_i \tilde{h}_i}} Z\Bigl[\mathcal{O}_1\Bigl(\frac{w_1}{\tau},\frac{\bar{w}_1}{\bar{\tau}}\Bigr)\cdots \mathcal{O}_n \Bigl(\frac{w_n}{\tau},\frac{\bar{w}_n}{\bar{\tau}}\Bigr)\Bigr]\Bigl(-\frac{1}{\tau},-\frac{1}{\bar\tau}\Bigr) \,. \label{eqn:ZOntauMT}
\end{multline}

To leverage these constraints to learn about the charged spectrum, we introduce chemical potentials into the partition function. We assume that the global symmetry group $G$ of the internal CFT is compact.\footnote{This is closely related to (and perhaps a consequence of) the compactness of the internal CFT~\cite{Harlow:2018tng,Benjamin:2020swg}, as well as to the conjectured compactness of gauge symmetries in quantum gravity~\cite{Banks:2010zn}.}
Let $J^a(w)$, $a=1,\ldots,n_L$ and $\tilde{J}^{\tilde{a}}(\bar{w})$, $\tilde{a} = 1,\ldots,n_R$ be the left and right moving conserved currents associated to the $U(1)^{n_L + n_R}$ maximal torus of $G$. Their OPEs take the general form
\begin{align}
J^a(w) J^b(0) &\sim - \frac{k^{a b}}{w^2} \,, & J^a(w) \tilde{J}^{\tilde{b}}(0)&\sim 0\,, & \tilde{J}^{\tilde{a}}(\bar{w}) \tilde{J}^{\tilde{b}}(0) &\sim - \frac{\tilde{k}^{\tilde{a} \tilde{b}}}{\bar{w}^2} \,.
\end{align}
Taking the currents to be Euclidean self-adjoint by convention, $J^a(w)^\dag = J^a(w)$, $\tilde{J}^{\tilde{a}}(\bar{w})^\dag = \tilde{J}^{\tilde{a}}(\bar{w})$, unitarity implies that the symmetric matrices $k^{a b}$ and $\tilde{k}^{\tilde{a} \tilde{b}}$ are real and positive-definite, so we normalize the currents such that
\begin{equation}
J^a(w) J^b(0) \sim - \frac{\delta^{a b}}{w^2} \,, \qquad \tilde{J}^{\tilde{a}}(\bar{w}) \tilde{J}^{\tilde{b}}(0) \sim - \frac{\delta^{\tilde{a} \tilde{b}}}{\bar{w}^2} \,. \label{eqn:JJOPE}
\end{equation}
Now consider the line operator
\begin{equation}
U_\Sigma(\mu, \tilde{\mu}) = e^{i \mu_a \oint_\Sigma J^a(w) d w - i \tilde{\mu}_{\tilde{a}} \oint_\Sigma \tilde{J}^{\tilde{a}}(\bar{w}) d \bar{w}}\,, \qquad \mu_a, \tilde{\mu}_{\tilde{a}} \in \mathbb{C} \,.
\end{equation}
For $\mu_a, \tilde{\mu}_{\tilde{a}} \in \mathbb{R}$ the operator $U_\Sigma(\mu, \tilde{\mu})$ is the symmetry operator associated to $U(1)^{n_L + n_R}$. Since the group is compact, we have
\begin{equation}
U_\Sigma(\mu, \tilde{\mu})=U_\Sigma(\mu + \rho, \tilde{\mu}+\tilde{\rho}) \,, \qquad \forall (\rho, \tilde{\rho}) \in \Gamma \,, \label{eqn:muperiod}
\end{equation}
where $\Gamma \cong \mathbb{Z}^{n_L + n_R}$ is the period lattice and $U(1)^{n_L + n_R} \cong \mathbb{R}^{n_L + n_R} / \Gamma$.\footnote{Although the conserved currents can be split into left and right movers, the period lattice generally \emph{does not} factor into left and right-moving components.}

Inserting $U_\Sigma(\mu,\tilde{\mu})$ on the constant-worldsheet-time circle $\Im w = 0$, we obtain the ``flavored'' partition function:
\begin{equation}
Z(\mu, \tau; \tilde{\mu}, \bar{\tau}) = \sum q^{h-\frac{c}{24}} y^{Q} \bar{q}^{\tilde{h}-\frac{\tilde{c}}{24}} \tilde{y}^{\tilde{Q}} \,, \qquad q \df e^{2\pi i \tau} \,,\;  y^{Q} \df e^{2\pi i \mu_a Q^a} \,,\;  \tilde{y}^{\tilde{Q}} \df e^{-2\pi i \tilde{\mu}_{\tilde{a}} \tilde{Q}^{\tilde{a}}} ,
\end{equation}
where $Q^a$, $\tilde{Q}^{\tilde{a}}$ are the $U(1)^{n_L + n_R}$ charges of the states and $\mu_a$, $\tilde{\mu}_{\tilde{a}}$ are the associated chemical potentials (with fugacities $y$, $\tilde{y}$).

The flavored partition function is a useful avatar for the charged spectrum, but it is not clear a priori how it transforms under modular transformations, since the $S$ transform $\tau \to -1/\tau$ changes the 1-cycle $\Sigma$ on which the symmetry operator is inserted, generating a twisted sector in place of the chemical potential (see figure~\ref{fig:SdualWS}).
\begin{figure}\centering
\includegraphics[width=.7\textwidth]{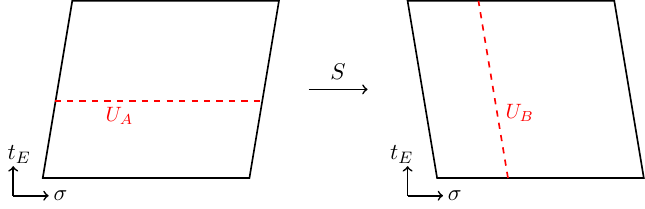}
\caption{The $S$ transform $\tau \to -1/\tau$ exchanges line operators wrapping the A and B cycles of the torus ($w\to w+2\pi$ and $w\to w+2\pi\tau$, respectively). Whereas symmetry operators inserted on the A cycle generate a chemical potential, on the B cycle they have a very different effect, creating a twisted sector. The flavored partition function thus fails to be modular invariant.}
\label{fig:SdualWS}
\end{figure}
Nevertheless, it has been argued that the flavored partition function transforms according to a simple, universal rule~\cite{Benjamin:2016fhe}:
\begin{equation}
Z(\mu,\tau;\tilde{\mu},\bar{\tau}) = Z(\mu,\tau+1;\tilde{\mu},\bar{\tau}+1) = e^{-\frac{\pi i}{\tau} \mu^2 +\frac{\pi i}{\bar{\tau}} \tilde{\mu}^2} Z\biggl(\frac{\mu}{\tau},-\frac{1}{\tau};\frac{\tilde{\mu}}{\bar{\tau}},-\frac{1}{\bar{\tau}}\biggr) \,, \label{eqn:ZmuMT}
\end{equation}
where $\mu^2 \df \mu_a \mu_b \delta^{a b}$, $\tilde{\mu}^2 \df \tilde\mu_{\tilde{a}} \tilde\mu_{\tilde{b}} \delta^{\tilde{a} \tilde{b}}$.

\subsubsection*{Proof of \eqref{eqn:ZmuMT}}

Since the formula~\eqref{eqn:ZmuMT} is crucial to our argument, we now provide an explicit and (to our knowledge) novel proof of it for completeness. This complements the more telegraphic argument given in~\cite{Benjamin:2016fhe}.

Consider an $n$-point function of (left-moving) conserved currents $J^a(w_i)$; since the currents are weight $(1,0)$ primaries,
\begin{multline}
Z[J^{a_1}(w_1) \cdots J^{a_n}(w_n)](\tau) = Z[J^{a_1}(w_1) \cdots J^{a_n}(w_n)](\tau+1) \\ = \frac{1}{\tau^n} Z\Bigl[J^{a_1}\Bigl(\frac{w_1}{\tau}\Bigr) \cdots J^{a_n}\Bigl(\frac{w_n}{\tau}\Bigr)\Bigr]\Bigl(-\frac{1}{\tau}\Bigr) \,, \label{eqn:JnMT}
\end{multline}
according to \eqref{eqn:ZOntauMT}, where we suppress $\bar{\tau}$'s henceforward for notational simplicity. Due to the OPE \eqref{eqn:JJOPE}, this $n$-point function is singular when the $w_i$ are not all distinct. These singularities can be removed by a normal ordering procedure, see appendix~\ref{app:NormalOrder} for a general discussion. Specifically, we consider the ``modular'' normal ordering:\footnote{Since the infinite double sum vanishes for $w_1=w_2$, this coincides with the conformal normal ordering $\Lcolon (\cdots) \Rcolon$ defined in~\S\ref{subsec:SugawaraOPEs} when the operators are coincident.}
\begin{align}
\Lcolon J^a(w_1) J^b(w_2) \Rcolon_\tau &\df J^a(w_1) J^b(w_2) + \frac{\delta^{a b}}{(w_1-w_2)^2} \nonumber \\
&\noeq \qquad +\delta^{a b} \sum_{(m,n)\ne (0,0)} \biggl(\frac{1}{[w_1-w_2+2\pi(m+n \tau)]^2} - \frac{1}{[2\pi (m+n\tau)]^2}\biggr) \nonumber \\
&= J^a(w_1) J^b(w_2) + \frac{\delta^{a b}}{4 \pi^2} \wp\Bigl(\frac{w_1-w_2}{2\pi}\Big|\tau\Bigr) \,,
\end{align}
where $\wp(u|\tau) \equiv \frac{1}{u^2} +  \sum_{(m,n)\ne (0,0)} \bigl[\frac{1}{(u+m+n\tau)^2} - \frac{1}{(m+n\tau)^2}\bigr]$ is the Weierstrass $\wp$ function. The infinite sum over poles removes not just the singularity where $w_1 \to w_2$, but also those where $w_1 \to w_2 + 2\pi$, $w_1 \to w_2 + 2\pi\tau$, etc., which arise due to the periods of the torus. The $w$-independent counterterms make the double sum absolutely convergent, which ensures modularity, $\wp(u|\tau) = \wp(u|\tau+1) = \frac{1}{\tau^2}\wp(u/\tau|-1/\tau)$ in addition to the manifest periodicity $\wp(u|\tau) = \wp(u+1|\tau) = \wp(u+\tau|\tau)$.

Since $\Lcolon J^a(w_1) J^b(w_2) \Rcolon_\tau$ is free from singularities, so too is the modular-normal-ordered $n$-point function:
\begin{equation}
Z\bigl[\Lcolon J^{a_1}(w_1) \cdots J^{a_n}(w_n)\Rcolon_\tau\bigr](\tau) \,, \label{eqn:ZmodNO}
\end{equation}
see the general proof given in appendix~\ref{app:NormalOrder}. The modular properties of $\frac{1}{4 \pi^2} \wp\Bigl(\frac{w_1-w_2}{2\pi}\Big|\tau\Bigr)$ match those of $J^a(w_1) J^b(w_2)$, so this modular-normal-ordered $n$-point function retains the modular transformation~\eqref{eqn:JnMT}, in particular
\begin{equation}
Z\bigl[\Lcolon (\mu_a J^a(0))^n \Rcolon_\tau\bigr](\tau) = Z\bigl[\Lcolon (\mu_a J^a(0))^n \Rcolon_\tau\bigr](\tau+1) = \frac{1}{\tau^n} Z\bigl[\Lcolon (\mu_a J^a(0))^n \Rcolon_\tau\bigr](-1/\tau) \,.
\end{equation}
These transformation rules are elegantly encoded in the generating function
\begin{equation}
\mathcal{Z}(\mu,\tau) \df Z\Bigl[ \Lcolon e^{2 \pi i \mu_a J^a(0)} \Rcolon_\tau \Bigr](\tau) \quad \text{where} \quad \mathcal{Z}(\mu,\tau+1) = \mathcal{Z}(\mu/\tau,-1/\tau) = \mathcal{Z}(\mu,\tau) \,. \label{eqn:modularCalZ}
\end{equation}

To relate $\mathcal{Z}(\mu,\tau)$ to the flavored partition function $Z(\mu,\tau)$, note that by construction~\eqref{eqn:ZmodNO} is an entire function of the insertion positions $w_i$ that is doubly periodic $w_i \cong w_i + 2\pi \cong w_i + 2\pi \tau$. Since there are no non-constant entire functions on the torus, such a function can only be a constant, and we can freely average each insertion over the A cycle, replacing currents $J^a(w)$ with charge operators $J_0^a = \frac{1}{2\pi} \int_0^{2\pi}\! J^a(w) d w$, whose eigenvalues are the charges $Q^a$. Note, however, that the normal-ordering contributions cannot be neglected, as the Weierstrass $\wp$ function has a non-zero average:\footnote{This can be derived by integrating the double-sum definition of $\wp$ term-by-term, then dropping terms for which the inner sum $\sum_{m = \infty}^\infty (\cdots)$ vanishes.}
\begin{equation}
\int_0^1 \wp(u|\tau) du = -\sum_{n=-\infty}^\infty \,\sideset{}{'}\sum_{m=-\infty}^\infty \frac{1}{(m+n\tau)^2} = -\frac{\pi^2}{3} E_2(\tau) \,, \label{eqn:WeierstrassPAvg}
\end{equation}
where the primed sum omits the term $m=0$ when $n=0$ and $E_2(\tau)$ is the holomorphic Eisenstein series of weight 2, which is a quasimodular form. Since the double sum is not absolutely convergent, the order of summation matters, which is related to the quasimodularity of $E_2(\tau)$.

Accounting for \eqref{eqn:WeierstrassPAvg}, one finds that, e.g.,
\begin{equation}
Z\bigl[\Lcolon J^{a}(0) J^{b}(0)\Rcolon_\tau\bigr](\tau) = Z\biggl[J_0^a J_0^b- \frac{\delta^{a b}}{12}  E_2(\tau) \biggr](\tau) \,.
\end{equation}
A similar calculation was done in~\cite{Benjamin:2016fhe}; since the left-hand-side is a weight-two modular form, the known quasimodular properties of $E_2(\tau)$ determine the modular transformation of $Z[J_0^a J_0^b ](\tau)$. While this can be argued to extend to higher orders in the charge operators in a universal way~\cite{Benjamin:2016fhe}, we can be more explicit as follows.

Define the ``zero-mode'' normal ordering:\footnote{This reduces to creation/annihilation normal-ordering $\lcirc (\cdots) \rcirc$ in the cylindrical limit $\tau \to i \infty$, see~\eqref{eqn:CAnormalorder}.}
\begin{align}
\lcirc J^a(w_1) J^b(w_2) \rcirc_\tau &\df J^a(w_1) J^b(w_2) + \delta^{a b} \sum_{n=-\infty}^\infty \sum_{m=-\infty}^\infty \frac{1}{[w_1-w_2+2\pi(m+n \tau)]^2}  \nonumber \\
&= J^a(w_1) J^b(w_2) + \frac{\delta^{a b}}{4 \pi^2} \wp^{(0)} \Bigl(\frac{w_1-w_2}{2\pi}\Big|\tau\Bigr) \,,
\end{align}
where $\wp^{(0)}(u|\tau) \equiv \wp(u|\tau) + \frac{\pi^2}{3} E_2(\tau)$ is the Weierstrass $\wp$ function shifted to have zero average along the A cycle, $\int_0^1 \wp^{(0)}(u|\tau) du = 0$. The $n$-point function of zero-mode-normal-ordered currents on the torus is still an entire, doubly periodic function---hence a constant---so we can replace the current operators with charge operators as before, but now the normal-ordering contributions drop out due to $\int_0^1 \wp^{(0)}(u|\tau) du=0$, leaving
\begin{equation}
Z\bigl[\lcirc J^{a_1}(0) \cdots J^{a_n}(0)\rcirc_\tau\bigr](\tau) = Z[J_0^{a_1} \cdots J_0^{a_n} ](\tau)\,.
\end{equation}
Thus, the flavored partition function can be written as
\begin{equation}
Z(\mu,\tau) = Z\Bigl[ \lcirc e^{2 \pi i \mu_a J^a(0)} \rcirc_\tau \Bigr](\tau)\,.
\end{equation}

Relating this to the modular function $\mathcal{Z}(\mu,\tau)$ is now a straightforward exercise in ``reordering'' the normal-ordered exponential. Using~\eqref{eqn:reorder} and
$\lcirc J^a(w_1) J^b(w_2) \rcirc_\tau = \Lcolon J^a(w_1) J^b(w_2) \Rcolon_\tau + \frac{\delta^{a b}}{12} E_2(\tau)$, one finds
\begin{multline}
\lcirc e^{2 \pi i \mu_a J^a(0)} \rcirc_\tau = \exp\biggl[\frac{E_2(\tau)}{24} \delta^{a b} \iint d w_1 d w_2 \frac{\delta}{\delta J^a(w_1)} \frac{\delta}{\delta J^b(w_2)} \biggr] \Lcolon e^{2 \pi i \mu_a J^a(0)} \Rcolon_\tau \\[3pt]
= \exp\biggl[-\frac{\pi^2}{6} E_2(\tau) \mu^2\biggr] \Lcolon e^{2 \pi i \mu_a J^a(0)} \Rcolon_\tau \,,
\end{multline}
where $\mu^2 \df \mu_a \mu_b \delta^{a b}$ as before. Thus,
\begin{equation}
Z(\mu,\tau) = \exp\biggl[-\frac{\pi^2}{6} E_2(\tau) \mu^2\biggr] \mathcal{Z}(\mu,\tau) \,.
\end{equation}
Using \eqref{eqn:modularCalZ} together with $E_2 (\tau) = E_2 (\tau + 1) = \frac{1}{\tau^2} E_2 (- 1 / \tau) + \frac{6 i}{\pi \tau}$, we finally obtain
\begin{equation}
Z(\mu,\tau) = Z(\mu,\tau+1) = e^{-\frac{\pi i}{\tau} \mu^2} Z\biggl(\frac{\mu}{\tau},-\frac{1}{\tau}\biggr) \,,
\end{equation}
which matches \eqref{eqn:ZmuMT} in the case $\tilde{\mu}_a = 0$.
Incorporating the right-moving chemical potentials is a straightforward exercise that we leave to the reader.

To summarize, our proof relies on the contrasting properties of two normal orderings, one used to construct a modular function $\mathcal{Z}(\mu,\tau)$ and the other to construct the flavored partition function $Z(\mu,\tau)$ itself. By relating the two normal orderings, we read off the modular transformation of the flavored partition function~\eqref{eqn:ZmuMT}.

\subsubsection*{Application of \eqref{eqn:ZmuMT}}

With~\eqref{eqn:ZmuMT} established, we can now argue rigorously for an infinite tower of charged states.
Combining with~\eqref{eqn:muperiod}, we obtain the quasiperiodicity condition~\cite{Heidenreich:2016aqi}
\begin{equation}
Z(\mu + \tau \rho, \tau, \tilde{\mu} + \bar{\tau} \tilde{\rho}, \bar{\tau}) = e^{-2\pi i \mu \rho - \pi i \rho^2 \tau + 2 \pi i \tilde{\mu}\tilde{\rho} + \pi i \tilde{\rho}^2 \bar\tau} Z(\mu, \tau, \tilde{\mu}, \bar{\tau}) \,, \qquad (\rho,\tilde{\rho}) \in \Gamma \,. \label{eqn:quasiperiod}
\end{equation}
This can be written more usefully in terms of $\hh \df h-\frac{1}{2} Q^2$, $\hht \df \tilde{h}-\frac{1}{2} \tilde{Q}^2$:
\begin{equation}
Z = \sum q^{\hh-\frac{c}{24}} \bar{q}^{\hht-\frac{\tilde{c}}{24}} q^{\frac{1}{2}Q^2} y^{Q} \bar{q}^{\frac{1}{2} \tilde{Q}^2} \tilde{y}^{\tilde{Q}} =  \sum q^{\hh-\frac{c}{24}} \bar{q}^{\hht-\frac{\tilde{c}}{24}} q^{\frac{1}{2}(Q+\rho)^2} y^{Q+\rho} \bar{q}^{\frac{1}{2} (\tilde{Q}+\tilde{\rho})^2} \tilde{y}^{\tilde{Q}+\tilde{\rho}} \,.
\end{equation}
In particular, matching the $q$-expansions on both sides, we conclude that the spectrum is invariant under the shift
\begin{equation}
(Q,\tilde{Q}) \to (Q,\tilde{Q}) + (\rho,\tilde{\rho}) \,,\qquad (\rho,\tilde{\rho}) \in \Gamma \,, \qquad \text{with $\hh$, $\hht$ held fixed.} \label{eqn:SpecFlow}
\end{equation}
Thus, starting with the state $\ket{1}$ corresponding to the identity operator, we obtain an infinite tower of charged states $\ket{\psi}$ with weights 
\begin{equation}
h = \frac{1}{2}Q^2, \qquad \tilde{h} = \frac{1}{2}\tilde{Q}^2, \qquad \text{for all charges} \quad (Q,\tilde{Q}) \in \Gamma \,, \label{eqn:psiweights}
\end{equation}
as first shown in~\cite{Heidenreich:2016aqi,Montero:2016tif}. Note that these states have the minimum possible weights given their charges, since if a lower weight state existed then by applying~\eqref{eqn:SpecFlow} in reverse we would obtain a neutral state of negative weight.\footnote{One can reach the same conclusion using the Sugawara construction, which implies the unitarity bounds $h \ge \frac{1}{2} Q^2$, $\tilde{h} \ge \frac{1}{2} \tilde{Q}^2$, see~\S\ref{subsec:unitarity}.} 

Several comments are in order. First, the charge lattice
\begin{equation}
\Gamma^\ast = \{ (Q, \tilde{Q})\; |\; \forall\,(\rho,\tilde\rho)\in\Gamma, \,Q\rho-\tilde{Q}\tilde{\rho} \in \mathbb{Z} \}
\end{equation}
is in general distinct from the period lattice $\Gamma$, but~\eqref{eqn:SpecFlow} implies that $\Gamma \subseteq \Gamma^\ast$. Likewise, $h - \tilde{h} \in \mathbb{Z}$ implies that $Q^2 - \tilde{Q}^2 \in 2 \mathbb{Z}$ for all $(Q,\tilde{Q}) \in \Gamma$, hence the period lattice $\Gamma$ is an even integral lattice; only in special cases is it self-dual, $\Gamma^\ast = \Gamma$. More generally, the finite abelian group $K = \Gamma^\ast / \Gamma$ plays a role similar to that of the level of a Kac-Moody algebra.

Second,~\eqref{eqn:SpecFlow} is an example of spectral flow, as follows. Insert $U_\Sigma(\mu, \tilde{\mu})$, $\mu_a,\tilde\mu_\alpha \in \mathbb{R}$, on a temporal circle as in figure~\ref{fig:SdualWS}, generating a twisted sector. When $(\mu,\tilde\mu)$ lie on the period lattice we recover the untwisted sector, but adiabatically increasing $(\mu,\tilde\mu)$ from zero to a point on the period lattice may induce a spectral flow; for instance, the ground state may adiabatically transform into an excited state.
In fact, a modular transformation relates this process to the quasiperiod~\eqref{eqn:quasiperiod} of the flavored partition function, so~\eqref{eqn:SpecFlow} describes the resulting spectral flow. This flow
 provides a physical mechanism to generate the entire tower of charged states from the vacuum.

Finally, the foregoing discussion is intimately connected to the Sugawara construction, to be discussed in the next subsection. In particular, $\hh$ and $\hht$ are the eigenvalues of the Virasoro operators $\hL_0$ and $\hLt_0$ associated to the ``remainder'' CFT left over when the conserved current portion is removed.

\bigskip

Therefore, per~\eqref{eqn:psimass}, any bosonic string theory contains an infinite tower of charged particles of mass
\begin{equation}
\frac{\alpha'}{4} m^2 = \frac{1}{2} \max(Q^2, \tilde{Q}^2) -1 \,, \qquad \text{for all $(Q,\tilde{Q}) \in \Gamma$.} \label{eqn:towermass}
\end{equation}
This does not yet prove the WGC (nor the stronger, sublattice WGC) as we have not shown that the particles in the tower are superextremal. In simple examples, such as toroidal compactifications of the heterotic string~\cite{Sen:1994eb} or orbifolds thereof~\cite{Heidenreich:2016aqi},
\begin{equation}
\frac{\alpha'}{4} M^2_{\text{BH}} \ge \frac{1}{2} \max(Q^2, \tilde{Q}^2) \,. \label{eqn:ExtHetOrbifold}
\end{equation}
The central result of our paper is that this holds more generally, as we will show in~\S\ref{subsec:selfrepulsive}.

\subsection{Long-range forces and the RFC}\label{subsec:bosonic-univ}

We now calculate the long-range forces between the particles in the tower~\eqref{eqn:towermass} that we have just shown to exist.

Using~\eqref{eqn:selfforceCFT2},
the long-range force on a particle $\Psi$ (with antiparticle $\bar{\Psi}$) can be calculated from three-point functions on the Riemann sphere of the form:
\begin{equation}
  \langle \bar{\Psi} \Psi \gamma^A \rangle, \langle \bar{\Psi} \Psi \mathfrak{g} \rangle,
  \langle \bar{\Psi} \Psi \Phi^I \rangle, \langle \gamma^A \gamma^B \mathfrak{g} \rangle,
  \langle \mathfrak{g}\mathfrak{g}\mathfrak{g} \rangle, \langle \Phi^I \Phi^J \mathfrak{g} \rangle \label{eqn:3ptNeeded}
\end{equation}
where $\gamma^A$, $\mathfrak{g}$ and $\Phi^I$ are
gauge boson, graviton, and massless\footnote{Since we are considering bosonic string theory, there is inevitably at least one tachyon (scalar with mass $m^2 < 0$) in the spectrum. Although interactions mediated by tachyons do not fall off exponentially like those mediated by scalars with $m^2>0$, they behave differently than those mediated by massless scalars, and we do not consider them here. See~\S\ref{sec:tachyons} for further discussion.} scalar vertex operators, respectively, and $\Psi$ is the vertex operator of a charged particle in the tower~\eqref{eqn:towermass}.

To evaluate these three-point functions, note that each of the vertex operators in question is a
tensor product of external and internal factors, hence the three-point functions
factor:
\begin{equation}
  V_i = V_i^{\mathrm{ext}} \otimes V_i^{\mathrm{int}} \qquad \Rightarrow \qquad
  \expect{ V_1 V_2 V_3 } = \expect{ V_1^{\mathrm{ext}} V_2^{\mathrm{ext}}
  V_3^{\mathrm{ext}} }_{\mathrm{ext}} \cdot \expect{ V_1^{\mathrm{int}}
  V_2^{\mathrm{int}} V_3^{\mathrm{int}} }_{\mathrm{int}} ,
\end{equation}
where we count the conformal gauge-fixing determinant ($bc$ ghosts) as part of the external CFT.
As the external three-point functions are by definition universal (i.e., they depend only on the external CFT), we need only compute the internal CFT three-point functions,
 fixing the universal factors by comparing to known results for bosonic string theory compactified on a torus (see appendix~\ref{app:toroidal}). Note that, although the entire closed string three-point amplitude does not depend on the positions of the three vertex operators, the internal factor $\expect{ V_1^{\mathrm{int}}
  V_2^{\mathrm{int}} V_3^{\mathrm{int}} }_{\mathrm{int}}$ \emph{will} depend on these positions. This dependence is universal, as it must cancel the position-dependence of $\expect{ V_1^{\mathrm{ext}} V_2^{\mathrm{ext}}
  V_3^{\mathrm{ext}} }_{\mathrm{ext}}$, so it is sufficient to compute the coefficient of this universal, position-dependent factor.

Per the discussion in~\S\ref{subsec:worldsheetEFT}, up to normalization the internal CFT vertex operators take the form
\begin{align}
\fg_{\text{int}} &= 1\,, &
\Phi^{0}_{\text{int}} &= 1\,, &
\Phi^i_{\text{int}} &= \varphi^i(z,\bar{z})\,, &
\gamma^A_{\text{int}} &= J^a(z), \tilde{J}^{\tilde{a}}(\bar{z})\,, &
\Psi_{\text{int}} &= \psi(z,\bar{z})\,,
\end{align}
where $\Phi^0$ denotes the dilaton and $\Phi^i$, $i>0$, denote the remaining massless scalars, arising from $(1,1)$ primaries $\varphi^i(z,\bar{z})$ of the internal CFT.

The dilaton does not kinetically mix with the other massless scalars, because
\begin{equation}
\expect{\Phi^{0}_{\text{int}}\, \Phi^i_{\text{int}}\, \fg_{\text{int}} } =  \expect{ \varphi^i } = 0 \,,
\end{equation}
since a sphere one-point function of a non-zero weight operator must vanish.\footnote{Scale invariance requires
$\langle \mathcal{O} (z, \bar{z}) \rangle_{S^2} \propto \frac{1}{z^h
\bar{z}^h}$ for $\mathcal{O}$ of weight $(h, \tilde{h})$, but this is not
translationally invariant unless $h = \tilde{h} = 0$.\label{fn:onepoint}} Thus, both the graviton and the dilaton make universal contributions to the long-range force, independent of the internal CFT. Comparing with~\eqref{eqn:gravdilaton}, we find
\begin{equation}
\mathcal{F}^{\fg + \Phi^0} = - \kappa_d^2 m m' \label{eqn:gravdilForce1}
\end{equation}
for the graviton/dilaton mediated force between string modes of mass $m$, $m'$. Thus,~\eqref{eqn:selfforceCFT2} becomes
\begin{equation}
\mathcal{F} = \mathcal{N}_J \frac{\expect{\bpsi \psi J^A} \expect{J J}^{-1}_{A B} \expect{J^B \bpsi' \psi'}}{\expect{\bpsi \psi} \expect{1}^{-1} \expect{\bpsi' \psi'}}
- \kappa_d^2 m m' 
- \mathcal{N}_{\varphi} \frac{\expect{\bpsi \psi \varphi^i} \expect{\varphi \varphi}^{-1}_{i j} \expect{\varphi^j \bpsi' \psi'}}{\expect{\bpsi \psi} \expect{1}^{-1} \expect{\bpsi' \psi'}} \,, \label{eqn:selfforceCFTinternal}
\end{equation}
where $J^A = (J^a, \tilde{J}^{\tilde{a}})$ represents both left and right-moving conserved currents, $\langle 1 \rangle$ is the sphere partition
function of the internal CFT, and $\mathcal{N}_J$, $\mathcal{N}_{\varphi}$ are normalizing factors determined by the external CFT, with a universal dependence on the positions of the vertex operators that cancels that of the internal CFT correlators.

Per the discussion
in \S\ref{subsec:worldsheetEFT}, the two- and three-point functions in~\eqref{eqn:selfforceCFTinternal} can equivalently
be expressed as inner products and matrix elements, respectively, giving\footnote{Here $\mathcal{N}_J'$ and $\mathcal{N}_{\varphi}'$ are again normalizing factors with a universal dependence on the positions of the $J$ and $\varphi$ insertions, respectively.  They are related to $\mathcal{N}_J$, $\mathcal{N}_{\varphi}$ by taking the $z_i \to \infty$ limits of various vertex operators $V_i(z_i,\bar{z}_i)$, stripping off appropriate powers of $z_i$ to be absorbed by the internal CFT limit \eqref{eqn:matrixElementCorrelator}.}
\begin{equation}
\mathcal{F} =  \mathcal{N}_J' \frac{\expect{\psi|J^A|\psi} \expect{J|J}^{-1}_{A B} \expect{\psi'|J^B|\psi'}}{\expect{\psi|\psi} \expect{1}^{-1} \expect{\psi'|\psi'}}
- \kappa_d^2 m m' 
- \mathcal{N}_{\varphi}' \frac{\expect{\psi|\varphi^i |\psi} \expect{\varphi|\varphi}^{-1}_{i j} \expect{\psi'|\varphi^j|\psi'}}{\expect{\psi|\psi} \expect{1}^{-1} \expect{\psi'|\psi'}}  \,. \label{eqn:selfforceCFTinternal2}
\end{equation}
Consider the first term. Using \eqref{eqn:matrixElementCorrelator} and conformal normal ordering, we find
\begin{equation}
  \langle J^a |J^b \rangle = \lim_{z \rightarrow \infty} (- z^2) \langle
  J^a (z) J^b (0) \rangle
  = \lim_{z \rightarrow \infty} \langle \delta^{a b} -
  z^2 \Lcolon J^a (z) J^b (0) \Rcolon \rangle = \delta^{a b} \langle 1 \rangle,
\end{equation}
since the sphere expectation value of the
normal-ordered product vanishes, see \S\ref{subsec:conformalNO}. 
By similar calculations, $\langle J^a | \tilde{J}^{\tilde{b}} \rangle = 0$
and $\langle \tilde{J}^{\tilde{a}} | \tilde{J}^{\tilde{b}} \rangle = \delta^{\tilde{a} \tilde{b}} \langle 1 \rangle$.
Conformal invariance dictates that $\langle \psi |J^a (z) | \psi \rangle
\propto \frac{1}{z}$, therefore applying Cauchy's integral theorem and the $J^a(z)$ mode expansion in radial quantization~\eqref{eqn:radialModeExp1}:
\begin{equation}
  \langle \psi | \oint \frac{\D z'}{2 \pi} J^a (z') | \psi \rangle = -
  \langle \psi |J_0^a | \psi \rangle = - Q^a \expect{\psi|\psi} \qquad
  \Rightarrow \qquad \langle \psi |J^a (z) | \psi \rangle = \frac{i}{z}
  Q^a  \expect{\psi|\psi}, \label{eqn:psiJpsi}
\end{equation}
where $Q^a$ is the left-moving charge of $\Psi$.\footnote{The $z$-dependence will be cancelled by the conformal ghosts / gauge-fixing determinant, so we can ignore it when computing the long-range forces, up to a universal factor.} Likewise $\langle \psi | \tilde{J}^{\tilde{a}} (\bar{z}) | \psi \rangle = - \frac{i}{\bar{z}}  \tilde{Q}^{\tilde{b}} \expect{\psi|\psi}$, and analogous expressions hold for $\Psi'$ involving the corresponding charges $Q^{a\prime}$, $\tilde{Q}^{\tilde{b}\prime}$. We conclude that the gauge-boson mediated long range force is proportional to
\begin{equation}
\mathcal{F}^{\gamma} \propto \delta_{a b} Q^a Q^{b\prime} + \delta_{\tilde{a} \tilde{b}} \tilde{Q}^{\tilde{a}} \tilde{Q}^{\tilde{b}\prime} \,,
\end{equation}
up to a universal factor. The overall normalization can be fixed by comparing with the known result for bosonic string theory compactified on a torus, \eqref{eqn:gaugeforce}:
\begin{equation}
\mathcal{F}^{\gamma} = \frac{2 \kappa_d^2}{\alpha'}  (\delta_{a b} Q^a Q^{b\prime} + \delta_{\tilde{a} \tilde{b}} \tilde{Q}^{\tilde{a}} \tilde{Q}^{\tilde{b}\prime}) . \label{eqn:gaugeForce1}
\end{equation}

Now consider the last term in~\eqref{eqn:selfforceCFTinternal2}.
First, note that a subset of the massless scalars arise from
$(1, 1)$ primaries of the form $\lambda^{a \tilde{b}} (z, \bar{z}) = J^a (z) 
\tilde{J}^{\tilde{b}} (\bar{z})$. These are orthogonal:
\begin{equation}
  \langle \lambda^{a \tilde{b}} | \lambda^{c \tilde{d}} \rangle = \lim_{z
  \rightarrow \infty} \langle |z|^4 J^a (z) \tilde{J}^{\tilde{b}} (\bar{z})
  J^c (0) \tilde{J}^{\tilde{d}} (0) \rangle = \delta^{a c} \delta^{\tilde{b}
  \tilde{d}} \langle 1 \rangle .
\end{equation}
Now suppose that $\chi (z, \bar{z})$ is another neutral modulus that is linearly independent from $\lambda^{a \tilde{b}}$. Without loss of generality, we can assume that $\chi$ is orthogonal
to $\lambda^{a \tilde{b}}$, so that
\begin{multline}
  0 = \langle \lambda^{a \tilde{b}} | \chi \rangle = \lim_{z \rightarrow
  \infty} | z |^4 \langle 1| J^a (z) \tilde{J}^{\tilde{b}} (\bar{z}) | \chi
  \rangle = \sum_{n, \tilde{n} > 0} \lim_{z \rightarrow \infty}  \langle 1|
  \frac{J^a_n}{z^{n - 1}} 
  \frac{\tilde{J}^{\tilde{b}}_{\tilde{n}}}{\bar{z}^{\tilde{n} - 1}} | \chi
  \rangle = \langle 1| J_1^a \tilde{J}_1^{\tilde{b}} | \chi \rangle, \\
  \Longrightarrow \qquad J_1^a \tilde{J}_1^{\tilde{b}} | \chi \rangle = 0 ,
  \label{eqn:chicons}
\end{multline}
where we use $\langle 1 | J_{-n}^a = \langle 1 | \tilde{J}_{-n}^a = 0$ for $n\ge 0$ since $J_{-n}^a = (J_n^a)^\dag$ acts to the left as a lowering operator.
Note that $\tilde{J}_1^{\tilde{b}} | \chi \rangle$ has weight $(1,0)$ and is therefore a left-moving conserved current. Since by assumption $J^a(z)$, $a=1,\cdots,n_L$ form a basis of left-moving conserved currents, we have
\begin{equation}
\tilde{J}_1^{\tilde{b}} | \chi \rangle = \Lambda^{\tilde{b}}_{\; a} J_{- 1}^a | 1 \rangle \label{eqn:chilower}
\end{equation}
for some coefficients $\Lambda^{\tilde{b}}_{\; a}$, where $J_{- 1}^a \ket{1}$ is the state corresponding to $J^a(z)$.\footnote{More precisely, if the worldsheet symmetry algebra $\mathfrak{g}$ is non-abelian then there will be additional conserved currents that are charged under the maximal torus that we are considering. However, these cannot appear in~\eqref{eqn:chilower} since $\chi$ is neutral by assumption.}
However, using
\eqref{eqn:chicons} this implies that
\begin{equation}
  0 = J_1^a (\tilde{J}_1^{\tilde{b}} | \chi \rangle) = \Lambda^{\tilde{b}}_{\;
  c} J_1^a J_{- 1}^c | 1 \rangle = \Lambda^{\tilde{b}}_{\; a} | 1 \rangle
  \qquad \Longrightarrow \qquad \Lambda^{\tilde{b}}_{\; a} = 0,
\end{equation}
using the algebra \eqref{eqn:LJalgebra}.
Thus $\tilde{J}_1^{\tilde{b}} | \chi \rangle = 0$, and by the same argument
$J_1^a | \chi \rangle = 0$. Since by assumption $\chi (z, \bar{z})$ is
neutral, $J_0 \ket{\chi} = \tilde{J}_0 \ket{\chi} = 0$. Therefore, by unitarity
\begin{equation}
 J_n^a | \chi \rangle = \tilde{J}_n^{\tilde{b}} | \chi \rangle = 0 \qquad
   \text{for} \qquad n \geqslant 0, \label{eqn:chiCond}
\end{equation}
i.e., $\chi (z, \bar{z})$ is a neutral \emph{current primary}.

By essentially the same argument that led to \eqref{eqn:psiJpsi}, we have:
\begin{equation}
  \langle \psi | \lambda^{a \tilde{b}} (z, \bar{z}) | \psi \rangle =
  \frac{1}{| z |^2} Q^a \tilde{Q}^{\tilde{b}} \langle 1 \rangle .
\end{equation}
This leaves only matrix elements of the form $\langle \psi | \chi (z, \bar{z}) |
\psi \rangle$ still to be calculated. At first, this seems very difficult,
since we have specified almost nothing about the modulus $\chi (z, \bar{z})$,
other than its orthogonality to the $J \tilde{J}$ moduli $\lambda^{a \tilde{b}}
(z, \bar{z})$.

In fact, $\langle \psi | \chi (z, \bar{z}) | \psi \rangle = 0$. To show this,
we employ the Sugawara construction (see, e.g.,~\cite{Goddard:1986bp}). As reviewed in appendix \ref{app:sugawara},
the stress tensor splits into two components:
\begin{equation}
  T (z) = T^j (z) + \hat{T} (z), \qquad T^j (z) = - \frac{1}{2} \delta_{a b} \Lcolon
  J^a (z) J^b (z) \Rcolon,
\end{equation}
each with associated Virasoro generators $L_m^j, \hat{L}_m$ that satisfy
their own Virasoro algebras and commute with each other. Using~\eqref{eqn:Ljmodes}, we compute:
\begin{equation}
  L^j_0 | \psi \rangle = \frac{1}{2} \sum_{p = - \infty}^{\infty} \delta_{a b}
  \lcirc J^a_{- p} J_p^b \rcirc  | \psi \rangle = \frac{1}{2} \delta_{a b} Q^a Q^b
  | \psi \rangle \qquad \Rightarrow \qquad h^j_{\psi} = \frac{1}{2}
  Q^2 ,
\end{equation}
where $J_n^a | \psi \rangle = 0$ for $n > 0$ since $| \psi \rangle$ is the
lowest-weight state of charge $Q_{\psi}^a$. Since $\ket{\psi}$ has the same weight $h_\psi = \frac{1}{2} Q^2$ under the full Virasoro algebra by \eqref{eqn:psiweights}, we conclude that $\hat{h}_\psi = h_\psi - h_\psi^j = 0$.
Doing the same computation for the right-movers, we find
\begin{equation}
(h^j_{\psi}, \tilde{h}^j_{\psi}) =
\frac{1}{2}  (Q^2, \tilde{Q}^2), \qquad (\hat{h}_{\psi},
\widetilde{\hat{h}}_{\psi}) = (0, 0) .
\end{equation}
 By comparison, using~\eqref{eqn:chiCond}
\begin{equation}
  L_0^j | \chi \rangle = \frac{1}{2} \sum_{p = - \infty}^{\infty} \delta_{a b}
  \lcirc J^a_{- p} J_p^b \rcirc  | \chi \rangle = 0 \qquad \Rightarrow \qquad h_{\chi}^j = 0 .
\end{equation}
Since $\ket{\chi}$ has weight $h_{\chi} = 1$ under the full Virasoro algebra, $\hat{h}_{\chi} = h_{\chi} - h_{\chi}^j = 1$. Accounting for the right-movers in the same way, we find
\begin{equation}
(h^j_{\chi}, \tilde{h}^j_{\chi}) = (0, 0), \qquad  (\hat{h}_{\chi}, \widetilde{\hat{h}}_{\chi}) = (1, 1) \,.
\end{equation}
By a similar calculation, $L^j_n | \psi
\rangle = L^j_n | \chi \rangle = 0$ for $n > 0$. Since $\ket{\psi}$, $\ket{\chi}$ are primaries under the full Virasoro algebra, this implies that $\hat{L}_n | \psi
\rangle = \hat{L}_n | \chi \rangle = 0$ for $n > 0$.

Thus, $\psi$ and $\chi$ are primaries with non-zero weight under $T^j$ and
$\hat{T}$, respectively, and no weight under the other component of the stress
tensor. Because the two components do not interact, the three-point function
$\langle \psi | \chi (z, \bar{z}) | \psi \rangle$ should be separately
invariant under conformal transformations induced by $T^j$ and $\hat{T}$.
Then, since $\hat{T}$ conformal transformations only act on $\chi$ and not
$\psi$, we must have $\langle \psi | \chi (z, \bar{z}) | \psi \rangle = 0$
because conformal invariance implies that sphere one-point functions of
non-zero weight operators must vanish, see footnote~\ref{fn:onepoint}.

For completeness, we now rederive this important fact explicitly using
operator algebra. First, note that by a standard argument
\begin{equation}
  \| \hat{L}_{- 1} | \psi \rangle \|^2 = \langle \psi | \hat{L}_1 \hat{L}_{-
  1} | \psi \rangle = \langle \psi | [\hat{L}_1, \hat{L}_{- 1}] | \psi \rangle
  = 2 \langle \psi | \hat{L}_0 | \psi \rangle = 0 \qquad \Rightarrow \qquad
  \hat{L}_{- 1} | \psi \rangle = 0,
\end{equation}
hence $\hat{L}_0, \hat{L}_{\pm 1}$ all annihilate $| \psi \rangle$. Next, we
rewrite the fact that $|\chi\rangle$ has weight $\hat{h}_{\chi} = 1$ in
operator language (in terms of $\chi (z, \bar{z})$):
\begin{align}
  \hat{L}_0 | \chi \rangle = | \chi \rangle \qquad \Longrightarrow \qquad \chi
  (0) & = \oint_{\Sigma} \frac{\D z'}{2 \pi i} z'  \hat{T} (z') \chi (0)
  \nonumber\\
  \Longrightarrow \qquad \chi (z, \bar{z}) & = \oint_{\Sigma'} \frac{\D
  z'}{2 \pi i}  (z' - z)  \hat{T} (z') \chi (z, \bar{z}), 
\end{align}
where in the second step we translate $\chi$ away from the origin. Here the
integration contour $\Sigma$ is a counterclockwise loop about $z' = 0$ and
$\Sigma'$ is a counterclockwise loop about $z' = z$. However, provided there
are no other operators along the circle $| z' | = | z |$, $\Sigma'$ can be
deformed into a pair of oppositely oriented circles at $| z' | = | z | \pm
\varepsilon$, see figure~\ref{fig:ContourDeform}.
\begin{figure}
\centering
\includegraphics[width=\textwidth]{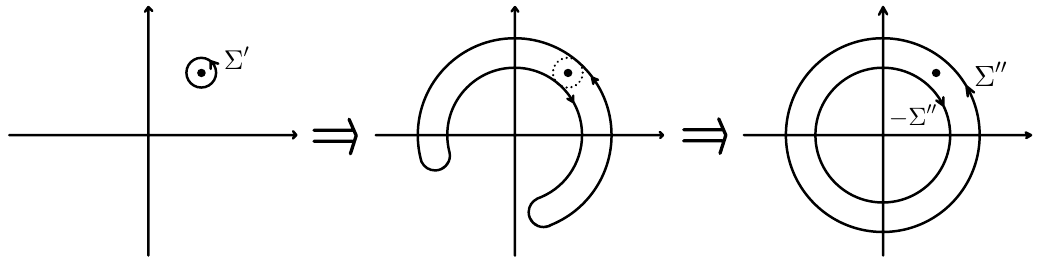}
\caption{Deforming a small loop about $\chi(z,\bar{z})$ into a pair of origin-centered circles of radius $|z| \pm \varepsilon$. A contour integration on the latter becomes a commutator in radial quantization.}\label{fig:ContourDeform}
\end{figure}
 Due to the implicit radial
time-ordering, this produces a commutator:
\begin{equation}
  \chi (z, \bar{z}) = \left[ \oint_{\Sigma''} \frac{\D z'}{2 \pi i}  (z' -
  z)  \hat{T} (z'), \chi (z, \bar{z}) \right] = [\hat{L}_0 - z \hat{L}_{- 1},
  \chi (z, \bar{z})], \label{eqn:chiident}
\end{equation}
where $\Sigma''$ is the counterclockwise circle $| z' | = | z |$ and we apply
the mode expansion~\eqref{eqn:radialModeExp2} in the second step.

Taking the $\langle \psi | \cdots | \psi \rangle$ matrix element of
\eqref{eqn:chiident}, we obtain:
\begin{equation}
  \langle \psi | \chi (z, \bar{z}) | \psi \rangle = \langle \psi | [\hat{L}_0
  - z \hat{L}_{- 1}, \chi (z, \bar{z})] | \psi \rangle = 0,
\end{equation}
as expected, where we use $\hat{L}_{- 1}^{\dag} = \hat{L}_1$ (see~\S\ref{subsec:unitarity}) and the fact
that $\hat{L}_0, \hat{L}_{\pm 1}$ annihilate $| \psi \rangle$.

Gathering all the pieces, we conclude that the long-range force mediated by the massless scalars other than the dilaton is proportional to
\begin{equation}
\mathcal{F}^{\Phi^i} \propto \delta_{a b} Q^a Q^{b\prime} \, \delta_{\tilde{a} \tilde{b}} \tilde{Q}^{\tilde{a}} \tilde{Q}^{\tilde{b}\prime} \,,
\end{equation}
up to a universal factor. Comparing with~\eqref{eqn:lambdaforce} to fix the normalization, we find
\begin{equation}
  \mathcal{F}^{\Phi^i} = - \frac{4 \kappa_d^2}{{\alpha'}^2} \, \frac{\delta_{a b} Q^a Q^{b\prime} \, \delta_{\tilde{a} \tilde{b}} \tilde{Q}^{\tilde{a}} \tilde{Q}^{\tilde{b}\prime}}{m m'} . \label{eqn:scalarForce1}
\end{equation}
Combining~\eqref{eqn:gravdilForce1}, \eqref{eqn:gaugeForce1}, \eqref{eqn:scalarForce1}, we find the net long-range force coefficient:
\begin{align}
  \mathcal{F} &= \frac{2 \kappa_d^2}{\alpha'}  (\delta_{a b} Q^a Q^{b\prime} + \delta_{\tilde{a} \tilde{b}} \tilde{Q}^{\tilde{a}} \tilde{Q}^{\tilde{b}\prime})
  -\frac{4 \kappa_d^2}{{\alpha'}^2} \, \frac{\delta_{a b} Q^a Q^{b\prime} \, \delta_{\tilde{a} \tilde{b}} \tilde{Q}^{\tilde{a}} \tilde{Q}^{\tilde{b}\prime}}{m m'}
   -\kappa_d^2 m m' \nonumber \\
  & = - \frac{4 \kappa_d^2}{{\alpha'}^2 m m'}  \biggl(
  \frac{\alpha'}{2} m m' - \delta_{a b} Q^a Q^{b\prime} \biggr) \biggl(
  \frac{\alpha'}{2} m m' - \delta_{\tilde{a} \tilde{b}} \tilde{Q}^{\tilde{a}} \tilde{Q}^{\tilde{b}\prime} \biggr) , \label{eqn:netForce1}
\end{align}
for any pair of states in the tower \eqref{eqn:towermass}.

By adding additional oscillator modes to the tower \eqref{eqn:towermass} either in the external directions or within the current algebra, one obtains a more general class of states with masses
\begin{equation}
\frac{\alpha'}{4} m^2 = \frac{1}{2} Q^2 +N -1 = \frac{1}{2} \tilde{Q}^2 + \tilde{N} -1 \,, \qquad \text{for $(Q,\tilde{Q}) \in \Gamma$, $N,\tilde{N}\ge 0$,} \label{eqn:towermass2}
\end{equation}
where the tower \eqref{eqn:towermass} consists of the lightest states of each charge.
The above reasoning is unaffected by the presence of these oscillator modes, so~\eqref{eqn:netForce1} applies to the states~\eqref{eqn:towermass2} as well. It will be useful to keep track of these more general states below. 

In particular, the self-force coefficient for one of these states works out to
\begin{equation}
\mathcal{F}_{\text{self}} = - \frac{16 \kappa_d^2}{{\alpha'}^2 m^2} (N-1) (\tilde{N} - 1) \,, \qquad \text{where $\tilde{N} - N = \frac{1}{2}(Q^2 - \tilde{Q}^2)$.}\label{eqn:netForce2}
\end{equation}
There are three different cases to consider, depending on the charges:
\begin{enumerate}
\item
If $|Q^2 - \tilde{Q}^2| > 2$, then $|N-\tilde{N}| \ge 2$ and the lightest level-matched state has either $N=0$, $\tilde{N} \ge 2$ or vice versa, implying that it is strictly self-repulsive, $\mathcal{F}_{\text{self}} > 0$.
\item
If $|Q^2 - \tilde{Q}^2| = 2$ then the lightest level-matched state has $N=0$, $\tilde{N}=1$ or vice versa, implying that it has vanishing self-force, $\mathcal{F}_{\text{self}} = 0$. 
\item
If $|Q^2 - \tilde{Q}^2| = 0$ then the lightest level-matched state has $N=\tilde{N}=0$,\footnote{This state, with mass $\frac{\alpha'}{4} m^2 = \frac{1}{2} Q^2 - 1 = \frac{1}{2} \tilde{Q}^2 - 1$, is a tachyon in parts of the moduli space.} implying that it is \emph{self-attractive}, $\mathcal{F}_{\text{self}} < 0$.
\end{enumerate}
Thus, the lightest state, with mass~\eqref{eqn:towermass}, can have a self-force of any sign, depending on $Q^2 - \tilde{Q}^2$. However, in all three cases, the next-to-lightest state, with mass
\begin{equation}
\frac{\alpha'}{4} m^2  = \frac{1}{2} \max(Q^2,\tilde{Q}^2) \,, \label{eqn:secondlightest}
\end{equation}
has vanishing self force. All states heavier than this have $N,\tilde{N} \ge 2$, hence they are self-attractive.

Since there is a state with vanishing or repulsive long-range self-force for every $(Q,\tilde{Q}) \in \Gamma$, we conclude that the sublattice RFC holds at string tree level.

\subsection{Self-repulsiveness vs. superextremality and the WGC} \label{subsec:selfrepulsive}

We now leverage the above results about long-range forces to fix the black hole extremality bound and show that the tower~\eqref{eqn:towermass} is (strictly) superextremal.

To understand the relationship between self-repulsiveness and superextremality, we begin by defining the latter precisely. 
A particle of charge $\vec{q}$ and mass $m$ is superextremal in the sense of~\eqref{eqn:WGCsketch} if $M = N m$ is no larger than the mass of the lightest black hole of charge $\vec{Q} = N \vec{q}$ as $N \to \infty$. Note that the inclusion of a large factor $N$ is necessary to avoid murky questions about sub-Planckian ``black holes'' ($m$ itself may be well below the Planck scale), whereas the $N\to \infty$ limit makes the definition unambiguous by removing any dependence on derivative corrections to the effective action. 

Let us assume, for technical reasons, that the black holes we are interested in are spherically symmetric.\footnote{In the simplest examples it can be proven that there are spherically symmetric black holes that are no heavier than any other black hole solution of the same charge, see, e.g.,~\cite{Gibbons:1982jg,Gibbons:1993xt}. Whether this holds in general is not established in the literature, to our knowledge.} For the purpose of studying such solutions, we can simplify the two-derivative effective action to~\cite{Heidenreich:2020upe}
\begin{equation} \label{eqn:gen-action-BH}
S = \int \!\dd^d x\, \sqrt{-g} \biggl(\frac{1}{2 \kappa_d^2} R -\frac{1}{2} G_{i j}(\phi) \nabla\phi^i\cdot\nabla\phi^j - \frac{1}{2} \ff_{a b}(\phi) F_2^a \cdot F_2^b \biggr) \,,
\end{equation}
where $\phi^i$ are the neutral moduli, $G_{i j}(\phi)$ is the metric on moduli space, and $F_2^a = d A_1^a$ are the vector fields (or a Cartan subalgebra in the non-abelian case).

Note that, in comparison with~\eqref{eqn:gen-action-EFT}, there is no scalar potential in~\eqref{eqn:gen-action-BH} since scalar fields that are not moduli have been omitted. This is because we are interested in parametrically heavy black holes $M_{\text{BH}} \to \infty$, which are accurately described by the leading-order terms in a derivative expansion. Since the scalar potential $V(\phi)$ is \emph{lower} order in the derivative expansion than the kinetic terms, the leading-order solution must extremize it, implying that only scalar fields with vanishing potential are switched on at leading order in the $M_{\text{BH}} \to \infty$ limit.\footnote{For the same reason, bosonic string tachyons can be ignored in the following analysis.}

Assuming the action~\eqref{eqn:gen-action-BH} and making an appropriate gauge choice,  spherically symmetric charge $Q_a$ solutions with an event horizon take the form
\begin{align}
  \dd s^2 &=  - e^{2 \psi(r)} f (r) \dd t^2 +
  e^{- \frac{2}{d - 3} \psi(r)}  \left[ \frac{\dd r^2}{f (r)} + r^2 \dd
  \Omega^2_{d - 2} \right] \,, &  f (r) &= 1 - \frac{r_h^{d - 3}}{r^{d - 3}}\,, \nonumber \\ 
  F^a_2 &= \frac{\ff^{a b} (\phi(r)) Q_b}{V_{d - 2}}  \frac{e^{2\psi(r)}}{r^{d - 2}} \dd t \wedge \dd r \,,  
\end{align}
where $r_h \geqslant 0$ is a constant, $V_{d - 2} = 2 \pi^{\frac{d - 1}{2}} / \Gamma ( \frac{d - 1}{2} )$ is the volume of a $(d-2)$-sphere, $\ff^{a b}(\phi)$ is the inverse of $\ff_{a b}(\phi)$, and $\psi(r)$, $\phi^i(r)$ solve a system of second-order ODEs plus a first-order constraint. 
In terms of $z \df \frac{1}{(d - 3) V_{d-2} r^{d - 3}}$,\footnote{Note that, in comparison to~\cite{Heidenreich:2020upe}, $z^{\text{(here)}} = z^{\text{(there)}}/V_{d-2}$.}
\begin{gather}
  k_N^{- 1} \frac{\dd}{\dd z} [f \dot{\psi}] = \mathfrak{f}^{a b} (\phi) Q_a Q_b e^{2
  \psi}\,, \qquad 
  \frac{\dd}{\dd z} [f \dot{\phi}^i] + f \Gamma^i_{\; j k} (\phi)  \dot{\phi}^j
  \dot{\phi}^k = \frac{1}{2} G^{i j} (\phi) \mathfrak{f}^{a b}_{, j} (\phi)
  Q_a Q_b e^{2 \psi} \,, \nonumber\\
  k_N^{- 1} \dot{\psi} (f \dot{\psi} + \dot{f}) + f G_{i j} (\phi) 
  \dot{\phi}^i \dot{\phi}^j = \mathfrak{f}^{a b} (\phi) Q_a Q_b e^{2 \psi} \,, \label{eqn:BHeqns}
\end{gather}
where dots denote $z$-derivatives, $\Gamma^i_{\; j k}(\phi)$ is the Levi-Civita connection for the metric on moduli space $G_{i j}(\phi)$, $G^{i j}(\phi)$ is its inverse, $k_N = \frac{d - 3}{d - 2} \kappa_d^2$ is the rationalized Newton's
constant (so that the gravitational force is $- \frac{k_N m m'}{V_{d - 2} r^{d -2}}$), and now $f=1-z/z_h$.

Black hole solutions are those that remain smooth at $r=r_h$. The ADM mass of such a solution is
\begin{equation}
M_{\text{BH}} = k_N^{- 1}  \biggl[ - \dot{\psi}_{\infty} + \frac{1}{2 z_h} \biggr] \,.
\end{equation}
To bound this mass, define the functional:
\begin{equation}
  \hat{S}[\psi,\phi] = \int_0^{z_h} \biggl[ \frac{1}{2 k_N} (f \dot{\psi} + \dot{f})^2
  + \frac{1}{2} f^2 G_{i j}(\phi)  \dot{\phi}^i  \dot{\phi}^j + \frac{1}{2} e^{2
  \psi} \mathfrak{f}^{a b} (\phi) Q_a Q_b \biggr] \dd z.
\end{equation}
Applying~\eqref{eqn:BHeqns}, we obtain 
\begin{equation}
  \hat{S} = k_N^{- 1} \biggl( \int_0^{z_h} \frac{\dd}{\dd z} \biggl[ \frac{1 + f}{2} f
  \dot{\psi} \biggr] \dd z + \frac{1}{2 z_h} \biggr) = k_N^{- 1}  \biggl( -
  \dot{\psi}_{\infty} + \frac{1}{2 z_h} \biggr) = M_{\text{BH}} .
\end{equation}
Now we observe that $\hat{S}$ can be rewritten as 
\begin{multline}
  \hat{S} =  \int_0^{z_h} \biggl( \frac{k_N}{2}  \biggl[ \frac{f \dot{\psi} +
  \dot{f}}{k_N} + e^{\psi} M (\phi) \biggr]^2 + \frac{1}{2} G_{i j} 
  [ f \dot{\phi}^i + e^{\psi} G^{i k}  M_{, k} ] [f
  \dot{\phi}^j + e^{\psi} G^{j l}  M_{, l}] \biggr) \dd z \\
   + \frac{1}{2}  \int_0^{z_h} e^{2 \psi}  \bigl[\mathfrak{f}^{a b} Q_a
  Q_b - k_N M (\phi)^2 - G^{i j} M_{, i}  M_{, j} \bigr] \dd z -
   f e^{\psi} M(\phi) \biggr|_0^{z_h},  \label{eqn:Bogomolnyi}
\end{multline}
for any function $M(\phi)$. Suppose that $M(\phi)$ is chosen to satisfy
\begin{equation}
  \mathfrak{f}^{a b} Q_a Q_b - k_N M (\phi)^2 - G^{i j} M_{, i}
  M_{, j} \geqslant 0, \label{eqn:Wcond}
\end{equation}
with $M(\phi_h)$ finite.  Then we find
\begin{equation}
  \hat{S} \geqslant -
   f e^{\psi} M(\phi) \biggr|_0^{z_h} = M(\phi_{\infty}),
\end{equation}
so we derive the bound $M_{\rm BH} \geqslant M(\phi_{\infty})$. 

Suppose there is a particle of mass $m(\phi)$ and charge $q_a$ that is self-repulsive throughout moduli space. Then, by definition
\begin{equation} \label{eqn:SRcond}
  \mathfrak{f}^{a b} q_a q_b - k_N m^2 - G^{i j}  \frac{\partial m}{\partial
  \phi^i}  \frac{\partial m}{\partial \phi^j} \geqslant 0,
\end{equation}
for all $\phi$, see~\eqref{eqn:selfforce}. Substituting $M(\phi) = N m(\phi)$ and $\vec{Q} = N \vec{q}$, we conclude that~\eqref{eqn:Wcond} is satisfied, hence $M_{\text{BH}} \geqslant M = N m$, and the particle is superextremal (i.e., $\frac{| \vec{q} |}{m} \geqslant \frac{| \vec{Q}_{\text{BH}} |}{M_{\text{BH}}}$). Thus, a particle that is self-repulsive throughout moduli space is superextremal.

A few comments are in order. First, the above argument relies on the fact that the self-repulsive particle exists throughout moduli space, since if $m(\phi)$ is undefined in some region that the black hole solution explores then the argument fails.

Second, while strictly speaking the argument requires the particle to be self-repulsive \emph{everywhere} in the moduli space, the argument still applies when~\eqref{eqn:SRcond} fails somewhere ``far away,'' provided that black hole solutions do not venture that far in the moduli space. This is relevant to our perturbative arguments, because we cannot calculate the long-range forces in strongly coupled regions of moduli space. However, since for black holes with only electric charge the attractor mechanism~\cite{Ferrara:1995ih,Cvetic:1995bj,Strominger:1996kf,Ferrara:1996dd,Ferrara:1996um} tends to flow \emph{away} from regions of strong coupling,\footnote{In $D=4$ or $D=5$, black holes may also carry magnetic charges. This statement and those that follow do not apply to such magnetic / dyonic black holes. However, since the string states we are studying have no magnetic charge, we need only study electrically-charged black holes to determine if they are superextremal.} this is unlikely to affect the extremality bound at weak coupling. More precisely, one can show that a spherically symmetric black hole with large $g_s$ near its horizon cannot exist when the asymptotic string coupling $(g_s)_{\infty}$ is sufficiently weak, i.e., such a solution obeys a sharp bound $(g_s)_{\infty} \ge g_s^{\text{crit}}$ where $g_s^{\text{crit}}$ depends on the spacetime dimension $D$ and the maximum value of the string coupling at which the string-tree-level effective action can be trusted~\cite{BHbound}. Thus, for weak enough string coupling such non-perturbative black holes cannot affect the extremality bound.

Finally, note that~\eqref{eqn:Wcond} only includes moduli derivatives $M_{,i}$, but the long-range force~\eqref{eqn:selfforce} receives contributions from all massless scalars, not just the moduli. However, since the inclusion of additional scalars can only increase the attractive self-force, a self-repulsive particle will still satisfy (but not saturate) \eqref{eqn:Wcond} when these scalars contribute to the long-range self-force, so the conclusion stands.

With these caveats in mind, we now consider the next-to-lightest states,~\eqref{eqn:secondlightest}, which as we have shown have vanishing long-range self-force. By the above reasoning these states are (super)extremal
\begin{equation}
\frac{\alpha'}{4} M_{\text{BH}}^2 \ge \frac{1}{2} \max(Q^2,\tilde{Q}^2) > \frac{1}{2} \max(Q^2,\tilde{Q}^2)-1 \,, \label{eqn:towerSuperext}
\end{equation}
implying that the tower~\eqref{eqn:towermass} is \emph{strictly} superextremal. This proves the sublattice WGC at tree level in bosonic string theory.

\bigskip

While not essential to proving the WGC, it is interesting to ask whether the first inequality in~\eqref{eqn:towerSuperext} can be saturated, i.e., what is the mass of an extremal black hole? To saturate the inequality, the completed squares in~\eqref{eqn:Bogomolnyi} must vanish. One can show that this requires $f=1$ ($r_h =0$)~\cite{Harlow:2022gzl}, so that
\begin{align}
\dot{\psi} &= - k_N e^{\psi} M (\phi)\,, & \dot{\phi}^i &= - e^{\psi} G^{i j}  M_{, j} \,. \label{eqn:fakeWeqns}
\end{align}
These Bogomol'nyi-like equations imply the equations of motion~\eqref{eqn:BHeqns} given a function $M(\phi)$ saturating~\eqref{eqn:Wcond}. Since they are first order, there is a unique solution for $z>0$ ($r<\infty$) given the boundary conditions at $z=0$, $\psi(r = \infty) = 0$ and $\phi^i(r = \infty) = \phi_\infty^i$, the latter being the values of the moduli in the chosen string vacuum.\footnote{This reproduces the ``fake superpotential'' method of obtaining solutions introduced in~\cite{Ceresole:2007wx,Andrianopoli:2007gt,Andrianopoli:2009je,Andrianopoli:2010bj,Trigiante:2012eb}.}

Thus, we obtain an extremal black hole of mass $M(\phi_\infty)$, \emph{provided that} the solution to~\eqref{eqn:fakeWeqns} is regular at the horizon $z=\infty$ ($r=0$). One can show that this is the case if and only if the $M(\phi)$ gradient flow ends at a critical point $\phi_h$ (i.e., an ``attractor point'') in the interior of moduli space with $M(\phi_h) > 0$. However, the function $M(\phi)$ defined by
\begin{equation}
\frac{\alpha'}{4} M^2 = \frac{1}{2} \max(Q^2, \tilde{Q}^2) \,, \label{eqn:Mtest}
\end{equation}
has no such critical points since it depends exponentially on the dilaton, so electrically charged extremal black holes do not exist in bosonic string theory.\footnote{Similar issues arise in superstring theories. There may of course be related objects with a horizon-like region with string-scale curvature, but these are not black holes according to the definition of~\S\ref{sec:intro}.}

Let us therefore rephrase the question: are there black holes in the theory with charge-to-mass ratio arbitrarily close to saturating the first inequality in~\eqref{eqn:towerSuperext}? To answer this question, it is convenient to consider solutions to~\eqref{eqn:fakeWeqns} with singular horizons, which nonetheless are the limit of a smooth family of subextremal black hole solutions~\cite{Heidenreich:2020upe,Harlow:2022gzl}. A necessary and sufficient condition for this is that $\dot{\psi} \le 0$ for all $z\ge 0$, see~\cite{Harlow:2022gzl}, appendix A. Per~\eqref{eqn:fakeWeqns}, this is ensured if $M(\phi) \ge 0$ along the $M(\phi)$ downwards gradient flow starting at $\phi_\infty^i$.

We now ask if the function $M(\phi)$ in~\eqref{eqn:Mtest}
has the required properties. First, note that $(1,1)$ primaries $\lambda^{a \tilde{b}} (z, \bar{z}) = J^a (z) 
\tilde{J}^{\tilde{b}} (\bar{z})$ are exactly marginal operators.\footnote{Turning them on has the effect of boosting the period lattice $\Gamma$, changing its decomposition into left and right-moving components.} This implies that the corresponding massless scalar fields are moduli. Thus, since the next-to-lightest tower~\eqref{eqn:secondlightest} has vanishing long-range self force and the massless scalars that contribute to this force are moduli, \eqref{eqn:Wcond} is indeed saturated. Moreover, the form of~\eqref{eqn:Mtest} implies that $M^2(\phi) > 0$ everywhere,\footnote{By contrast, if we only had $M^2(\phi) \ge 0$ then $M(\phi)$ could pass between the positive and negative branches of the square root at the loci where $M^2(\phi) = 0$, violating $M(\phi) > 0$. This explains how the lightest tower in cases with $|Q^2 - \tilde{Q}^2| = 2$ can have zero self-force and yet be strictly superextremal.} so by choosing the positive square root, $M(\phi) > 0$ everywhere. By the above discussion, this is enough to conclude that the solution to~\eqref{eqn:fakeWeqns} is the limit of a family of subextremal black hole solutions,
i.e., the extremality bound is precisely
\begin{equation}
\frac{\alpha'}{4} M^2_{\text{BH}} \ge \frac{1}{2} \max(Q^2, \tilde{Q}^2) \,,
\end{equation}
just like in toroidal orbifolds~\eqref{eqn:ExtHetOrbifold}.

\section{Summary and Discussion}\label{sec:final}

In this paper, we have proven the Weak Gravity Conjecture (WGC) in all perturbative, bosonic string theories in dimensions $D\ge 6$, building on the partial arguments given in~\cite{ArkaniHamed:2006dz,Heidenreich:2016aqi}. Our proof proceeded in three steps. First, we used modular invariance of the internal CFT to fix the modular transformation of the flavored partition function (following~\cite{Benjamin:2016fhe}), implying that the charged spectrum is quasiperiodic (following~\cite{Montero:2016tif,Heidenreich:2016aqi}). Starting with the identity operator of the internal CFT (which generates the graviton), this implies that there is a tower of charged particles of mass~\cite{Heidenreich:2016aqi}:
\begin{equation}
\frac{\alpha'}{4} m^2 = \frac{1}{2} \max(Q^2, \tilde{Q}^2) + N -1 \,, \qquad N \in \mathbb{Z}_{\ge 0}, \label{eqn:towerSummary}
\end{equation}
for all charges $(Q,\tilde{Q})$ that lie on the \emph{period lattice} $\Gamma$, an even lattice that is a finite-index sublattice of the charge lattice $\Gamma_Q = \Gamma^\ast$. We emphasize that, although the mass formula~\eqref{eqn:towerSummary} resembles that of a toroidal (Narain) compactification of the bosonic string, this formula applies to \emph{any} bosonic string theory, including compactifications on arbitrary Ricci-flat manifolds, non-geometric compactifications, or indeed the theory described by any internal CFT with appropriate central charges. Unlike the Narain case,~\eqref{eqn:towerSummary} does not necessarily describe every string mode, and in particular says nothing about modes with charges not on the period lattice, $(Q,\tilde{Q}) \in \Gamma^\ast \setminus \Gamma$.

Next, we computed the long-range forces between the modes~\eqref{eqn:towerSummary}. To do so, we wrote the force in terms of three-point functions, most of which can be computed (up to universal factors depending only on the external CFT) using the current algebra of the internal CFT in a straightforward fashion. Non-trivially, we showed that the  remaining three-point functions---which describe the coupling of the modes~\eqref{eqn:towerSummary} to massless scalars other than the dilaton and $J \tilde{J}$ moduli---vanish. To prove this fact, we used the Sugawara construction to separate the parts of the internal CFT associated to the conserved currents from the remainder. Fixing the universal factors that we did not calculate by comparing with a reference theory (Narain compactification), we found the tree-level self-force coefficient for the modes~\eqref{eqn:towerSummary}:
\begin{equation}
  \mathcal{F} = - \frac{4 \kappa_d^2}{{\alpha'}^2 m^2}  \biggl(
  \frac{\alpha'}{2} m^2 - Q^2 \biggr) \! \biggl(
  \frac{\alpha'}{2} m^2 - \tilde{Q}^2 \biggr) \,.
\end{equation}
An analogous formula was derived for heterotic Narain compactification in~\cite{Heidenreich:2019zkl} (see also appendix~\ref{app:toroidal} for bosonic Narain compactification) but our argument shows that it holds for \emph{any} bosonic string theory, a new result of our paper. In particular, this implies that the $N=1$ modes in~\eqref{eqn:towerSummary} have zero long-range self force at string tree level.

Finally, we showed that a particle that is self-repulsive throughout moduli space is superextremal (following~\cite{Heidenreich:2019zkl,Harlow:2022gzl}). Since the $N=1$ modes in~\eqref{eqn:towerSummary} have zero long-range self force, this implies that they are superextremal:
\begin{equation}
\frac{\alpha'}{4} M_{\text{BH}}^2 \ge \frac{1}{2} \max(Q^2,\tilde{Q}^2) > \frac{1}{2} \max(Q^2,\tilde{Q}^2)-1 \,,
\end{equation}
hence the $N=0$ modes in~\eqref{eqn:towerSummary} are strictly superextremal. Thus, the tree-level spectrum contains a strictly superextremal particle for every charge $(Q,\tilde{Q})$ on the period lattice $\Gamma_{\text{ext}} = \Gamma$. This proves that the tree-level spectrum satisfies the strict, sublattice WGC, which in turn implies the Ooguri-Vafa version of the mild WGC, completing our proof.

While many aspects of our proof have appeared in previous literature, we have attempted to make this paper as self-contained as possible, reviewing and improving arguments for the essential facts, both new and old. For instance, we have found a novel, explicit argument deriving the modular transformation of the flavored partition function, reproducing the results of~\cite{Benjamin:2016fhe}.

\subsection{Loop corrections} \label{sec:loops}

As our proof directly constrains only the tree-level spectrum, it is interesting to consider how it might be affected by string loop corrections. These could have several important effects, such as correcting the masses of the particles in the tower~\eqref{eqn:towerSummary}, generating finite decay widths, or correcting the extremality bound itself.

First, we note that since the lightest particle of each charge $(Q,\tilde{Q})$ is \emph{strictly} superextremal by an amount proportional to the string scale, any corrections with a perturbative expansion in positive powers of the string coupling will be too small to make these particles subextremal when the string coupling is small, $g_s \ll 1$.

This leaves only loop corrections that qualitatively alter the situation to be considered. These may, for instance, lift the moduli space and/or generate a cosmological constant. Lifting some or all of the moduli has the effect of strengthening the extremality bound, making the WGC easier to satisfy~\cite{Heidenreich:2015nta,Harlow:2022gzl}. Strictly speaking, from the perspective of the mild WGC as defined in~\S\ref{sec:intro}, this is an abrupt effect: generating even a small moduli potential discontinuously alters the physics of parametrically large black holes, since when they are sufficiently large the potential will begin to dominate the derivative expansion. However, especially from the perspective of the sublattice WGC, a renormalized notion of extremality may be more natural to consider~\cite{Heidenreich:2017sim}. Then, depending on the scale we are interested in, we use black holes of different sizes to normalize the extremality bound, in which case a perturbatively generated moduli potential will have a perturbative effect on the extremality bound. In that case, by the reasoning described in the previous paragraph, the corrections cannot make particles in the superextremal tower subextremal when the string coupling is sufficiently small.

When string loop corrections generate a cosmological constant (CC), the situation is less clear. To define the WGC with a CC, the notion of ``extremality'' should be referenced to black holes that are large in Planck units but small in cosmological units~\cite{Nakayama:2015hga,Heidenreich:2016aqi,Harlow:2022gzl},\footnote{This is imprecise; a precise formulation of the WGC in the presence of a CC is not yet known.} since cosmological-scale black holes often have masses that increase with higher powers of their charges. Moreover, since the loop-generated CC scales with a positive power of $g_s$, a dilaton tadpole inevitably results, causing the tree-level vacuum to roll.\footnote{If the CC is positive, the tree-level vacuum rolls downhill towards zero string coupling (where the CC vanishes). If the CC is negative, the tree-level vacuum rolls downhill towards infinite string coupling, at least until perturbation theory breaks down.} In this case, we need to consider black holes in a time-dependent cosmological background. If the potential is shallow and the black hole is not too large, these can be approximated by asymptotically flat black holes. Thus, at weak string coupling there is once more an intermediate mass range of black holes that behave similarly to their tree-level counterparts, and if we use these black holes to define a suitable generalization of the WGC to such rolling backgrounds, the conjecture will be satisfied.

In general, complications such as approximate moduli with a small\-/but\-/nonzero potential and the presence of a (not too large) cosmological constant require us to think carefully about how the WGC is defined. (By contrast, the definitions in~\S\ref{sec:intro} apply to the idealized scenario where the moduli space is exact and the cosmological constant vanishes.) We expect that when this is done correctly, perturbatively small corrections that introduce such complications will not have a large effect on whether the appropriately defined WGC bound is satisfied, and thus we expect our tree-level results to extend to perturbative, bosonic string theory at sufficiently weak string coupling. Further exploring the correct statement of the WGC in the presence of these complications is an interesting topic for future research.

Note that, by contrast, the tree-level bosonic string \emph{does not} satisfy a strict form of the sublattice RFC, since there are no strictly self-repulsive particles in the cases $Q^2 - \tilde{Q}^2 \in \{0,\pm 2\}$. The effect of loop corrections on the self-force of these particles is another interesting topic for the future.

\subsection{Tachyons} \label{sec:tachyons}

Throughout our paper, we have ignored the obvious issue that all bosonic string theories have one or more tachyons in their tree-level spectra, rendering the vacuum unstable. In part, our present work is intended to set up a more meaningful generalization to superstring theory~\cite{BMSuperstring}, where tachyons can be avoided. However, it is also possible that bosonic string theory can be interpreted as an \emph{unstable} critical point in the landscape of QGTs (or, more precisely as a nearly flat hilltop, due to the one-loop dilaton tadpole). From this viewpoint, one can ask whether the WGC also constraints such unstable critical points in the same way that it constrains vacua in the QGT landscape. Our analysis suggests that it does.

One might wonder whether the presence of tachyons requires us to rethink the definition of the WGC, as in~\S\ref{sec:loops}. At least formally, this is not the case, because we can imagine ``integrating out'' the tachyons in much the same way that we integrate out massive fields, generating higher derivative operators that have a subleading effect in the derivative expansion. The only differences between tachyons and massive particles in such a procedure are an extra set of minus signs.

\subsection{Expanding the scope of the proof}

As noted in~\S\ref{sec:intro}, our proof (and its superstring generalization~\cite{BMSuperstring}) is limited to the electric NSNS gauge sector, since string modes are only charged under this sector. In the case of closed bosonic string theory, this excludes objects that are magnetically charged under the Kalb-Ramond two form $B_2$ as well as those that are magnetically charged under the same 1-form NSNS gauge fields whose electric charges we have been studying. Such objects---which can be particles in $D=5$ and $D=4$, respectively---correspond to target space solitons. To study them in a general bosonic string theory we would need a prescription for assembling a new worldsheet CFT describing the appropriate soliton background, starting with the internal CFT specifying the bosonic string theory in question. This is an interesting topic for future research.

In the type II superstring case, there are also RR gauge fields to consider, for which the charged states are D-branes. To study the D-brane spectrum systematically, we would need to systemically understand the possible boundary conditions allowed by the internal CFT. We defer further discussion of this interesting open problem to~\cite{BMSuperstring}.

\section*{Acknowledgments}

We thank M.~Montero 
for helpful discussions and M.~Etheredge and T.~Rudelius for comments on the manuscript. We also thank our anonymous referee for identifying several places where we glossed over important details in earlier drafts.
 BH was supported by National Science Foundation grants PHY-1914934, PHY-2112800 and PHY-2412570.
 ML was supported in part by the National Research Foundation of Korea (NRF) Grants 2021R1A2C2012350 and RS-2024-00405629.
BH thanks the Yau Mathematical Sciences Center and the Kavli Institute for the Physics and Mathematics of the Universe for hospitality during the inception of this work. This work was initiated at Aspen Center for Physics, which is supported by National Science Foundation grant PHY-1607611.

\appendix

\section{Holomorphic normal ordering}\label{app:NormalOrder}

Let $\mathcal{O}_1, \mathcal{O}_2, \mathcal{O}_3 \ldots$ be a collection of
holomorphic operators, and suppose that the singular parts of their mutual
OPEs on the complex plane are proportional to the identity
operator
\begin{equation}
  \mathcal{O}_i (z) \mathcal{O}_j (0) \sim \mathfrak{f}_{i j} (z) \cdot
  \boldsymbol{1}.
\end{equation}
Here the singular behavior of $f_{i j} (z) = f_{j i} (- z)$ as $z \rightarrow
0$ matches that of the OPE, whereas the finite part is arbitrary.

A \textbf{normal-ordering prescription} $N [\cdots]$ is a prescription for
removing the OPE singularity
\begin{equation}
  N [\mathcal{O}_i (z_1) \mathcal{O}_j (z_2)] \equiv \mathcal{O}_i (z_1)
  \mathcal{O}_j (z_2) - N_{i j} (z_1, z_2),
\end{equation}
where the \textbf{contractions} $N_{i j} (z_1, z_2) = N_{j i} (z_2, z_1)$,
also written as $\wick[offset=0.92em]{\c \cO_i(z_1) \c \cO_j(z_2)}$, are specified functions chosen such that $N_{i j} (z_1, z_2)$ is
analytic for $z_1 \neq z_2$ and $N_{i j} (z_1, z_2) -\mathfrak{f}_{i j} (z_1 -
z_2)$ is analytic as $z_1 \rightarrow z_2$. As a result, $N [\mathcal{O}_i
(z_1) \mathcal{O}_j (z_2)]$ is an entire function of $z_1$ and $z_2$. More
precisely, this holds as an operator equation, i.e., $\langle N [\mathcal{O}_i
(z_1) \mathcal{O}_j (z_2)] \rangle$ is an entire function, whereas $\langle
\cdots N [\mathcal{O}_i (z_1) \mathcal{O}_j (z_2)] \cdots \rangle$ is only
singular (as a function of $z_{1, 2}$) when $z_1$ or $z_2$ approaches the
insertion point of another operator.

We then define the normal-ordered product of $n$ operators $\mathcal{O}_1,
\ldots, \mathcal{O}_n$ as
\begin{equation}
  N [\mathcal{O}_1 (z_1) \cdots \mathcal{O}_n (z_n)] \equiv \mathcal{O}_1
  (z_1) \cdots \mathcal{O}_n (z_n) + \sum \text{subtractions},
  \label{eqn:NOsubtract}
\end{equation}
where the subtractions consist of all possible ways of replacing one or more
pairs of operators in the product $\mathcal{O}_1 (z_1) \cdots \mathcal{O}_n
(z_n)$ with $-\wick[offset=0.92em]{\c \cO_i(z_1) \c \cO_j(z_2)} = - N_{i j} (z_i, z_j)$. Normal-ordered products can be built-up
iteratively, since
\begin{multline}
  N [\mathcal{O}_1 (z_1) \cdots \mathcal{O}_n (z_n)] = N [\mathcal{O}_1 (z_1)
  \cdots \mathcal{O}_k (z_k)] N [\mathcal{O}_{k + 1} (z_{k + 1}) \cdots
  \mathcal{O}_n (z_n)] \\ + \sum \text{cross-subtractions}, \label{eqn:NOcross}
\end{multline}
where the cross-subtractions pair operators in $N [\mathcal{O}_1 (z_1) \cdots
\mathcal{O}_k (z_k)]$ exclusively with those in $N [\mathcal{O}_{k + 1} (z_{k
+ 1}) \cdots \mathcal{O}_n (z_n)]$. In particular, consider the $k = 2$ case:
\begin{equation}
  N [\mathcal{O}_1 (z_1) \cdots \mathcal{O}_n (z_n)] = N [\mathcal{O}_1 (z_1)
  \mathcal{O}_2 (z_2)] N [\mathcal{O}_3 (z_3) \cdots \mathcal{O}_n (z_n)] +
  \sum \text{cross-subtractions} .
\end{equation}
By construction, $N [\mathcal{O}_1 (z_1) \mathcal{O}_2 (z_2)]$ is non-singular
as $z_1 \rightarrow z_2$, provided that the other insertion points $z_3,
\ldots, z_n$ are separated from $z_1, z_2$. Likewise, the cross-subtractions
are non-singular under the same assumptions. Therefore, $N [\mathcal{O}_1
(z_1) \cdots \mathcal{O}_n (z_n)]$ is non-singular as $z_1 \rightarrow z_2$ so
long as $z_{\ell} \neq z_{1, 2}$ for all $\ell > 2$. The same argument applies
to any pair of variables $z_i, z_j$, hence $N [\mathcal{O}_1 (z_1) \cdots
\mathcal{O}_n (z_n)]$ is non-singular (as an operator equation) (1) when all
the variables are distinct and (2) when a single pair coincides.

Could $N [\mathcal{O}_1 (z_1) \cdots \mathcal{O}_n (z_n)]$ have singularities
at higher codimension? This is forbidden by \textbf{Hartogs's extension theorem}, which states that a holomorphic function of
$N > 1$ complex variables has no isolated (non-removable) singularities.\footnote{This is not the strongest version of the theorem, but is sufficient for our purposes.} Essentially, Hartogs's extension theorem implies that the singular locus of a holomorphic function of several variables cannot have complex codimension two or higher, because then there would be an isolated singularity on a two-dimensional slice transverse to the singular locus. To make this concrete, consider the following:
\begin{lemma}
A function $f(z_1, \ldots, z_n)$ of $n$ complex variables that is analytic when its arguments are distinct and when at most two coincide (with all others distinct) is an entire function. \label{lem:Hartog}
\end{lemma}
\begin{proof}
Assuming that $f(z_1,\ldots,z_n)$ is \emph{not} an entire function, it has at least one singular point $(\hat{z}_1, \ldots, \hat{z}_n)$. Define the holomorphic function:
\begin{equation}
F(x,y) \equiv f(\hat{z}_1 + (1+i) x-y, \ldots, \hat{z}_n + (n^2+i n^3) x  - n y )\,. \label{eqn:Fslice}
\end{equation}
To understand the singularity structure of $F(x,y)$, consider
\begin{equation}
z_p - z_q = \hat{z}_p - \hat{z}_q  + (p-q) [p+q +i (p^2+pq+q^2)] x - (p-q) y \,.
\end{equation}
Thus, $z_p = z_q$ along the line $y = [p+q +i (p^2+pq+q^2)] x + \frac{\hat{z}_p - \hat{z}_q}{p-q}$. By the assumed properties of $f(z_1,\ldots,z_n)$, $F(x,y)$ can only be singular where two or more of these lines intersect. Since $p+q +i (p^2+pq+q^2) \ne \tilde{p}+\tilde{q} +i (\tilde{p}^2+\tilde{p}\tilde{q}+\tilde{q}^2)$ for $(p,q) \ne (\tilde{p},\tilde{q}), (\tilde{q},\tilde{p})$, no two of these lines have the same slope---indeed, the slice~\eqref{eqn:Fslice} was chosen to have this property---hence they intersect at isolated points in $\mathbb{C}^2$. Since its potential singularities are isolated, see figure~\ref{fig:FxySingularities}, Hartog's extension theorem implies that $F(x,y)$ is an entire function. As this contradicts the assumption that $z_i = \hat{z}_i$ ($x=y=0$) was a singular point of $f$, $f(z_1,\ldots,z_n)$ must be an entire function.
\end{proof}
\noindent As shown previously, the normal-ordered product $N[\mathcal{O}_1(z_1) \cdots  \mathcal{O}_n (z_n)]$ is analytic when at most two of its arguments coincide, hence it is an entire
function (as an operator equation) by the above lemma.
\begin{figure}
\centering
\includegraphics[width=0.6\textwidth]{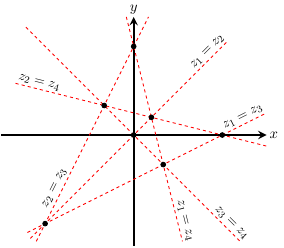}
\caption{A schematic picture of the potential singularities of $F(x,y)$---depicted in $\mathbb{R}^2$ rather than the actual domain $\mathbb{C}^2$---where the red dashed lines represent the (complex) lines $z_i = z_j$, and the black dots represent their intersections. By assumption $F(x,y)$ is analytic away from the intersections. Since the lines all have distinct slopes, these intersections are isolated points in $\mathbb{C}^2$, so Hartog's extension theorem implies that $F(x,y)$ is an entire function.}\label{fig:FxySingularities}
\end{figure}

By appropriate modifications to the contractions, the same method can be
applied to non-trivial Riemann surfaces. For instance, working a cylinder $w
\cong w + 2 \pi$ or a torus $w \cong w + 2 \pi \cong w + 2 \pi \tau$ in place
of the complex plane, the $\mathcal{O}_i (w_1) \mathcal{O}_j (w_2)$ OPE will
be singular when $w_i \rightarrow w_j + 2 \pi m$ for $m \in \mathbb{Z}$ (on
the cylinder) or when $w_i \rightarrow w_j + 2 \pi (m + n \tau)$ for $m, n \in
\mathbb{Z}$ (on the torus). Choosing a contraction $N_{i j} (w_1, w_2)$ that
matches all of these singularities, the proof that $N [\mathcal{O}_1 (w_1)
\cdots \mathcal{O}_n (w_n)]$ is entire is essentially unchanged.

Anti-holomorphic operators can be normal-ordered in the same way as above.
Since the OPEs of holomorphic and anti-holomorphic operators with each other
are non-singular, this can be extended to general products of holomorphic and
anti-holomorphic operators.\footnote{It should be possible to relax the
assumption that the singular part of the OPE is proportional to the identity
operator but we will not need this.}

We now collect a few more useful identities. Normal-ordering an
unsubtracted operator product is done using contractions in
place of subtractions:
\begin{equation}
  \mathcal{O}_1 (z_1) \cdots \mathcal{O}_n (z_n) = N [\mathcal{O}_1 (z_1)
  \cdots \mathcal{O}_n (z_n)] + \sum \text{contractions},
  \label{eqn:NOcontract}
\end{equation}
where the contractions consist of all possible ways of replacing one or more
pairs of operators in the product $\mathcal{O}_1 (z_1) \cdots \mathcal{O}_n
(z_n)$ with the contraction $\wick[offset=0.92em]{\c \cO_i(z_i) \c \cO_j(z_j)} = N_{i j} (z_i, z_j)$. More generally, one can translate between
two normal orderings $\tilde{N} [\mathcal{O}_1 (z_1) \mathcal{O}_2 (z_2)] = N
[\mathcal{O}_1 (z_1) \mathcal{O}_2 (z_2)] + \Delta_{i j} (z_1, z_2)$ using the
``recontraction'' $\Delta_{i j} (z_1, z_2) = \tilde{N}_{i j} (z_1, z_2) - N_{i
j} (z_1, z_2)$ (which is an entire function). One has:
\begin{equation}
  \tilde{N} [\mathcal{O}_1 (z_1) \cdots \mathcal{O}_n (z_n)] = N
  [\mathcal{O}_1 (z_1) \cdots \mathcal{O}_n (z_n)] + \sum
  \text{recontractions}, \label{eqn:NOrecontract}
\end{equation}
where the recontractions consist of all possible ways of replacing one or more
pairs of operators in the product $N[\mathcal{O}_1 (z_1) \cdots \mathcal{O}_n(z_n)]$ with the recontraction $\Delta_{i j} (z_1, z_2)$.
As in, e.g.,~\cite{Polchinski:1998rq}, the equations (\ref{eqn:NOsubtract}), (\ref{eqn:NOcross}),
\eqref{eqn:NOcontract}, \eqref{eqn:NOrecontract} can all be rewritten using
exponentials of functional derivatives. For instance, \eqref{eqn:NOrecontract} becomes:
\begin{equation}
  \tilde{N} [(\cdots)] = \exp \biggl[ \frac{1}{2}  \sum_{i, j} \iint \D z
  \D z' \Delta_{i j} (z, z')  \frac{\delta}{\delta \mathcal{O}_i (z)} 
  \frac{\delta}{\delta \mathcal{O}_j (z')} \biggr] N [(\cdots)], \label{eqn:reorder}
\end{equation}
where $(\cdots)$ is the expression to be re-ordered, $\frac{\delta
\mathcal{O}_i (z)}{\delta \mathcal{O}_j (z')} = \delta_i^j \delta (z - z')$,
and the integrals are defined formally using $\int \D z f (z)
\delta (z - z') = f (z')$ (rather than as contour integrals).

\subsection{Conformal normal ordering} \label{subsec:conformalNO}

If the holomorphic operators $\mathcal{O}_i(z)$ to be normal ordered are conformal primaries of weight $h_i>0$ then their OPEs take the general form
\begin{equation}
  \mathcal{O}_i (z) \mathcal{O}_j (0) \sim -\frac{k_{i j}}{z^{h_i + h_j}} ,
\end{equation}
for constants $k_{i j} = (-1)^{h_i+h_j} k_{j i}$, where we assume as before that the singular terms are proportional to the identity operator. Then it is natural to define the normal ordering:
\begin{equation} \label{eqn:confNormalOrder}
\Lcolon \mathcal{O}_i (z_1) \mathcal{O}_j (z_2) \Rcolon \equiv \mathcal{O}_i (z_1) \mathcal{O}_j (z_2) + \frac{k_{i j}}{(z_1 - z_2)^{h_i + h_j}} ,
\end{equation}
which has the fortuitous property that $\Lcolon \mathcal{O}_i (z_1) \mathcal{O}_j (z_2) \Rcolon$ transforms in the same way as $\mathcal{O}_i (z_1) \mathcal{O}_j (z_2)$ under global conformal (M\"obius) transformations $z \to \frac{a z + b}{c z+d}$.

This \emph{conformal normal ordering} is particularly useful on the Riemann sphere, where the expectation value of a conformal-normal-ordered product vanishes:
\begin{equation}
\langle \Lcolon \mathcal{O}_1(z_1) \cdots \mathcal{O}_n(z_n) \Rcolon \rangle_{S^2} = 0 \,.\label{eqn:NOexpectedvalue}
\end{equation}
This is because the left-hand side $f(z_1,\ldots,z_n)$ is an entire function by construction, whereas applying the global conformal transformation $z' = 1/z$:\
\begin{equation}
\langle \Lcolon \mathcal{O}_1(z_1) \cdots \mathcal{O}_n(z_n) \Rcolon \rangle_{S^2} = (z_1')^{2 h_1} \cdots (z_n')^{2 h_n} \langle \Lcolon \mathcal{O}_1(z_1') \cdots \mathcal{O}_n(z_n') \Rcolon\rangle_{S^2}\,.
\end{equation}
Thus, since $h_i > 0$, $f(z_1,\ldots,z_n) \to 0$ as $z_i \to \infty$ ($z_i' \to 0$) for any $i$. Per Cauchy's integral formula, an entire function that vanishes asymptotically vanishes identically, establishing~\eqref{eqn:NOexpectedvalue}.

\section{Review of the Sugawara construction} \label{app:sugawara}

In this appendix, we review the Sugawara construction (see, e.g.,~\cite{Goddard:1986bp}), whereby the
stress-energy tensor $T$ of a CFT with a conserved current $J$ is split into a
piece constructed from the current itself plus a remainder whose OPE with $J$
is non-singular. We focus on the abelian version of the construction, which is
all that is need for this paper.

\subsection{Local operators and OPEs}\label{subsec:SugawaraOPEs}

Suppose for simplicity that there is a single,
left-moving\footnote{Right-moving conserved currents can be handled in exactly
the same fashion as the left-moving ones, so we will say nothing more about
them.} conserved current $J$. The $J J$ OPE takes the general form $J (z) J
(0) \sim - \frac{k}{z^2}$ for some constant $k$, where unitarity requires $k >
0$. Without loss of generality, we rescale $J$ to set $k = 1$ henceforward, so
that
\begin{align}
  T (z) T (0) &\sim \frac{c}{2 z^4} + \frac{2}{z^2} T (0) + \frac{1}{z} \partial T (0), \nonumber \\
   T (z) J (0) &\sim \frac{1}{z^2} J (0) + \frac{1}{z}
  \partial J (0), & J (z) J (0) &\sim - \frac{1}{z^2},
\end{align}
using the fact that $J (z)$ is a $(1, 0)$ primary.

Using the contraction $\wick[offset=0.92em]{\c J (z_1) \c J (z_2)} \equiv - \frac{1}{(z_1 - z_2)^2}$, we define the conformal
normal-ordered product $\Lcolon J (z_1) \ldots J (z_n) \Rcolon$ of current operators as in
\S\ref{subsec:conformalNO}. This removes the OPE singularities, allowing us to define local operators $\Lcolon J^n (z) \Rcolon$ for any $n
\geqslant 0$. With some foresight, we decompose the stress tensor as
\begin{align}
  T (z) &= T^j (z) + \hat{T} (z), & \text{where} \quad T^j (z) &\equiv -
  \frac{1}{2} \Lcolon J^2 (z) \Rcolon, & \hat{T} (z) &\equiv T (z) - T^j (z) .
\end{align}
Normal ordering, we obtain
\begin{align}
  T^j (z) J (0) &= - \frac{1}{2} \Lcolon J^2 (z) \Rcolon J (0) = - \frac{1}{2} \Lcolon J^2(z) J(0) \Rcolon - \Lcolon \wick{J (z) \c J (z) \c J (0)} \Rcolon   \nonumber\\
  & = - \frac{1}{2} \Lcolon J^2 (z) J (0) \Rcolon + \frac{1}{z^2} J (z) \sim \frac{1}{z^2} J (0) + \frac{1}{z} \partial J (0) \sim T (z) J (0), 
\end{align}
implying that the $\hat{T} J$ OPE is non-singular, $\hat{T} (z) J (0) \sim 0$.
This in turn implies that $\hat{T} (z) \Lcolon J (z_1) \ldots J (z_n) \Rcolon$ is
non-singular as $z \rightarrow z_i$ for any $1 \leqslant i \leqslant n$ (with
$z_j \neq z_i$ for $j \neq i$), hence by Hartogs's extension theorem the OPE
$\hat{T} (z) \Lcolon J (z_1) \ldots J (z_n) \Rcolon$ is completely non-singular, see lemma~\ref{lem:Hartog}.

Using normal ordering, $\hat{T} (z) \Lcolon J^2 (0) \Rcolon \sim 0$, and the $T T$ OPE, the
$T^j$, $\hat{T}$ OPEs are straightforward to calculate:
\begin{align}
  T^j (z) T^j (0) &\sim \frac{1}{2 z^4} + \frac{2}{z^2} T^j (0) + \frac{1}{z}
  \partial T^j (0), & \hat{T} (z) T^j (0) &\sim 0, \nonumber \\
  \hat{T} (z) 
  \hat{T} (0) &\sim \frac{c - 1}{2 z^4} + \frac{2}{z^2}  \hat{T} (0) +
  \frac{1}{z} \partial \hat{T} (0) .
\end{align}
Thus, at the level of the OPE, $T^j$ and $\hat{T}$ function like decoupled
stress tensors for a current sector of central charge $c^j = 1$ and a remainder of central charge $\hat{c} = c - 1$.

This construction readily generalizes to $N$ abelian currents, for which
\begin{align}
  T (z) T (0) &\sim \frac{c}{2 z^4} + \frac{2}{z^2} T (0) + \frac{1}{z} 
  \partial T (0) , \nonumber \\
   T (z) J^a (0) &\sim \frac{1}{z^2} J^a (0) +
  \frac{1}{z} \partial J^a (0), & J^a (z) J^b (0) &\sim - \frac{\delta^{a
  b}}{z^2}, \label{eqn:TJopes}
\end{align}
after normalizing the currents as before, where $a, b = 1, \ldots N$. We then
have
\begin{align}
  T^j (z) &= - \frac{1}{2} \delta_{a b} \Lcolon J^a (z) J^b (z) \Rcolon, &
  \text{where} \qquad \wick{\c J^a (z_1) \c J^b (z_2)} &\equiv - \frac{\delta^{a b}}{(z_1 - z_2)^2} .
\end{align}
Proceeding as before, one finds:
\begin{align}
  T^j (z) T^j (0) &\sim \frac{N}{2 z^4} + \frac{2}{z^2} T^j (0) + \frac{1}{z}
  \partial T^j (0), & T^j (z) J^a (z) &\sim \frac{1}{z^2} J^a (0) +
  \frac{1}{z} \partial J^a (0), \nonumber\\ \hat{T} (z) T^j (0) &\sim 0 , & J^a (z) J^b (0) &\sim - \frac{\delta^{a b}}{z^2},  \nonumber\\
  \hat{T} (z)  \hat{T} (0) &\sim \frac{c - N}{2 z^4} + \frac{2}{z^2}  \hat{T}
  (0) + \frac{1}{z} \partial \hat{T} (0), & \hat{T} (z) J^a (0) &\sim 0\,.  \label{eqn:ThatTjJopes}
\end{align}
Thus, the OPEs factor into a current sector with effective stress tensor $T^j$
and central charge $c^j = N$, as well as a residual sector with effective
stress tensor $\hat{T}$ and central charge $\hat{c} = c - N$.

\subsection{Mode expansions}
\label{sec:modeexpansion}

Working in radial quantization, we consider the mode expansions:\footnote{Note that for the right-moving currents we use the mode expansion $\tilde{J}^a (\bar{z}) = -i \sum_n
  \frac{\tilde{J}_n^a}{\bar{z}^{n + 1}}$, which corresponds to the cylindrical mode expansion $\tilde{J}^a (\bar{w}) = \sum_n J_n^a e^{-i n \bar{w}}$.}
\begin{equation}
  T (z) = \sum_n \frac{L_n}{z^{n + 2}}, \qquad J^a (z) = i \sum_n
  \frac{J_n^a}{z^{n + 1}} . \label{eqn:radialModeExp1}
\end{equation}
By a standard calculation, applying these mode expansions to the OPEs
\eqref{eqn:TJopes} leads to the current algebra
\begin{align}
  [L_m, L_n] &= (m - n) L_{m + n} + \frac{c}{12} m (m^2 - 1) \delta_{m, - n}, \nonumber \\
  [L_m, J_n^a] &= - n J^a_{m + n}, & [J^a_m, J^b_n] &= m \delta^{a b} \delta_{m, - n} . \label{eqn:LJalgebra}
\end{align}
By the same calculation, applying the mode expansions
\begin{equation}
  T^j (z) = \sum_n \frac{L_n^j}{z^{n + 2}}, \qquad \hat{T} (z) = \sum_n
  \frac{\hat{L}_n}{z^{n + 2}}, \label{eqn:radialModeExp2}
\end{equation}
to \eqref{eqn:ThatTjJopes}, we obtain
\begin{align}
  [L_m^j, L_n^j] &= (m - n) L^j_{m + n} + \frac{N}{12} m (m^2 - 1) \delta_{m, - n}, & [L^j_m, J_n^a] &= - n J^a_{m + n},  \nonumber\\
  [\hat{L}_m, L^j_n]&= 0, & [J^a_m, J^b_n] &= m \delta^{a b} \delta_{m, - n}, \nonumber\\
  [\hat{L}_m, \hat{L}_n] &= (m - n)  \hat{L}_{m + n} + \frac{c - N}{12} m (m^2 - 1) \delta_{m, - n}, &  [\hat{L}_m, J_n] &= 0\,.  \label{eqn:LhatJalgebra}
\end{align}

Note that $L^j_n$ is built from the $J_m$ modes. To see how, note that for $|
z_1 | > | z_2 |$, the time-ordered product of two currents has the expansion
\begin{align}\label{eqn:JJOPEradial}
  {J^a}  (z_1) J^b (z_2) &= - \sum_{n_1, n_2} \frac{J_{n_1}^a J_{n_2}^b}{z_1^{n_1 + 1} z_2^{n_2
  + 1}} = - \sum_{n_1, n_2} \frac{\lcirc J_{n_1}^a J_{n_2}^b \rcirc }{z_1^{n_1 + 1}
  z_2^{n_2 + 1}}  - \sum_{n_1 > n_2} \frac{[J_{n_1}^a, J_{n_2}^b]}{z_1^{n_1 + 1} z_2^{n_2 + 1}} 
   \nonumber\\
  &=  \lcirc J^a(z_1) J^b(z_2) \rcirc - \frac{\delta^{a b}}{z_1^2}  \sum_{n = 0}^{\infty} n
  \frac{z_2^{n - 1}}{z_1^{n - 1}} = \lcirc J^a(z_1) J^b (z_2) \rcirc - \frac{\delta^{a b}}{(z_1 - z_2)^2}, 
\end{align}
where $\lcirc \cdots \rcirc$ denotes the normal-ordering of creation/annihilation operators.
Thus, $\lcirc J^a(z_1) J^b (z_2) \rcirc = \Lcolon J^a(z_1) J^b (z_2) \rcolon$ in radial quantization, and so
\begin{align}
  T^j (z) &= - \frac{1}{2} \delta_{a b}\,  \lcirc J^a(z) J^b (z) \rcirc = \frac{1}{2}  \sum_{n_1, n_2} \frac{\delta_{a b}}{z^{n_1 + n_2
  + 2}}  \lcirc J_{n_1}^a J_{n_2}^b  \rcirc\,, \nonumber\\ \Rightarrow \qquad L_n^j &= \frac{1}{2} \sum_{p = -
  \infty}^{\infty} \delta_{a b}\,  \lcirc J_{n - p}^a J_p^b  \rcirc . \label{eqn:Ljmodes}
\end{align}
As a cross-check, substituting $L_n = \hat{L}_n + L_n^j = \hat{L}_n +
\frac{1}{2}  \sum_{p = - \infty}^{\infty} \delta_{a b}\, \lcirc J_{n - p}^a J_p^b  \rcirc$ into \eqref{eqn:LJalgebra} and simplifying using the $J J$
algebra, one eventually obtains \eqref{eqn:LhatJalgebra}, as expected.

\paragraph{Cylindrical quantization}

One can also write the mode expansions~\eqref{eqn:radialModeExp1}, \eqref{eqn:radialModeExp2} in cylindrical quantization (in terms
of $w \cong w + 2 \pi$, with $z = e^{- i w}$):
\begin{align}
  T (w) &= - \sum_n L_n e^{i n w} + \frac{c}{24}, &
  J^a (w) &= \sum_n J_n^a e^{i n w}, \nonumber\\
  T^j (w) &= - \sum_n L_n^j e^{i n w} + \frac{N}{24}, &
  \hat{T} (w) &= -\sum_n \hat{L}_n e^{i n w} + \frac{c - N}{24},  \label{eqn:cylindermodes}
\end{align}
where the extra terms in the mode expansions of $T (w)$, $T^j (w)$, and
$\hat{T} (w)$ versus \eqref{eqn:radialModeExp1}, \eqref{eqn:radialModeExp2}
are due to their non-tensorial conformal transformations, fixed by the
$\frac{1}{z^4}$ term in their OPE with $T$.

The $\frac{N}{24}$ term in the mode expansion of
$T^j (w)$ can also be obtained by comparing conformal normal ordering and
creation/anihilation operator normal ordering in cylindrical quantization. In
particular, for $\Im w_1 > \Im w_2$, one finds:
\begin{align}
  {J^a}  (w_1) J^b (w_2) &=  \sum_{n_1, n_2} \! e^{i n_1 w_1 + i n_2 w_2} J_{n_1}^a J_{n_2}^b \nonumber\\
  &= \sum_{n_1, n_2}\! e^{i n_1 w_1 + i n_2 w_2} \lcirc J_{n_1}^a J_{n_2}^b \rcirc + \sum_{n_1 > n_2}\!\! e^{i n_1 w_1 + i n_2 w_2}  [J_{n_1}^a, J_{n_2}^b] \nonumber\\
  &=  
 \lcirc J^a(w_1) J^b(w_2) \rcirc - \frac{\delta^{a b}}{4 \sin^2 \bigl( \frac{w_1 - w_2}{2} \bigr)}, 
\end{align}
so that
\begin{equation}
  \lcirc J^a(w_1) J^b (w_2) \rcirc = \Lcolon J^a (w_1) J^b (w_2) \Rcolon + \delta^{a b}  \biggl[ \frac{1}{4 \sin^2 \left( \frac{w_1 - w_2}{2} \right)} - \frac{1}{(w_1 - w_2)^2} \biggr]. \label{eqn:CAnormalorder}
\end{equation}
In the limit $w_1 \rightarrow w_2$, this becomes $\lcirc J^a(w) J^b (w) \rcirc = \Lcolon J^a (w) J^b (w) \Rcolon + \frac{1}{12} \delta^{a b}$. Therefore,
\begin{align}
  T^j (w) &= - \frac{1}{2} \delta_{a b} \Lcolon J^a (z) J^b (z) \Rcolon 
    = - \frac{1}{2} \delta_{a b} \lcirc J^a(w) J^b(w) \rcirc + \frac{N}{24}  \nonumber\\
  &= - \frac{1}{2}  \sum_{n_1, n_2} e^{i (n_1 + n_2) w} \delta_{a b} \lcirc J_{n_1}^a J_{n_2}^b \rcirc + \frac{N}{24}
    = - \sum_n L_n^j e^{i n w} + \frac{N}{24}, 
\end{align}
as expected, where we used \eqref{eqn:Ljmodes}.

\subsection{Unitarity bounds} \label{subsec:unitarity}

Since they arise from self-adjoint operators in the Lorentzian theory, $T(w)$ and $J^a(w)$ are Euclidean self-adjoint
\begin{align}
T(w)^\dag &= T(w) \,, & J^a(w)^\dag &= J^a(w) \,,
\end{align}
in cylindrical quantization. Using the mode expansions~\eqref{eqn:cylindermodes}, \eqref{eqn:Ljmodes} and $w^\dag = w$, one finds
\begin{align}
L_m^\dag &= L_{-m}\,, & (J_m^a)^\dag &= J^a_{-m}\,, & (L^j_m)^\dag &= L^j_{-m}\,, & \hat{L}_m^\dag &= \hat{L}_{-m}\,,
\end{align}
so the operators $T^j(w)$ and $\hat{T}(w)$ are also Euclidean self-adjoint.

Much like unitarity requires the weight $L_0 |\psi\rangle = h |\psi\rangle$ of a non-zero state $|\psi\rangle \ne 0$ to be non-negative, $h \ge 0$, there are unitarity bounds on the weights $h^j$, $\hat{h}$ of a state under the Sugawara components $T^j(z)$, $\hat{T}(z)$, as follows:
\begin{statement}
A state $|\psi\rangle \ne 0$ of definite $\hat{T}(z)$ weight $\hat{L}_0 |\psi\rangle = \hat{h}|\psi\rangle$ satisfies $\hat{h} \ge 0$.
\end{statement}

\begin{proof}
If $|\psi\rangle$ is a $\hat{T}(z)$ conformal primary, $\hat{L}_{1} |\psi\rangle = 0$, then
\begin{equation}
\lVert \hat{L}_{-1} |\psi\rangle \rVert^2 = \langle \psi | \hat{L}_{1} \hat{L}_{-1} | \psi \rangle = \langle \psi | [\hat{L}_{1}, \hat{L}_{-1}] | \psi \rangle = 2 \langle \psi | \hat{L}_0 | \psi \rangle = 2 \hat{h} \langle \psi | \psi \rangle ,
\end{equation}
 so $\hat{h} = \lVert \hat{L}_{-1} |\psi\rangle \rVert^2 / (2 \langle \psi|\psi\rangle) \ge 0$.
 
More generally, consider the ladder of states $| \psi_p \rangle = \hat{L}_1^p | \psi \rangle$, $p\ge 0$, with weights $\hat{h}_p = \hat{h} - p$. Assuming a spectrum bounded from below, this ladder must have a bottom rung, i.e., there exists $P\ge 0$ such that $| \psi_P \rangle \ne 0$ and $| \psi_p \rangle = 0$ for $p > P$. Then by construction $| \psi_P\rangle$ is a $\hat{T}(z)$ conformal primary, and $\hat{h}_P = \hat{h} - P \ge 0$ by the preceeding argument, implying $\hat{h} \ge P \ge 0$.
\end{proof}

\begin{statement}
A state $|\psi\rangle \ne 0$ of definite charge $J^a_0 |\psi\rangle = Q^a |\psi\rangle$ and $T^j(z)$ weight $L^j_0 |\psi\rangle = h^j |\psi\rangle$ satisfies $h^j \ge \frac{1}{2} \delta_{a b} Q^a Q^b$.
\end{statement}
 
\begin{proof}
Since $(J_n^a)^\dag = J^a_{-n}$, the operator
\begin{equation}
L_0^j = \frac{1}{2} \delta_{a b} J_0^a J_0^b + \sum_{n=1}^\infty \delta_{a b} J_{-n}^a J_n^b = \frac{1}{2} \delta_{a b} J_0^a J_0^b + \sum_{n=1}^\infty \sum_{a=1}^N (J_{n}^a)^\dag J_n^a
\end{equation}
is non-negative, and
\begin{equation}
h^j = \frac{\langle \psi | L_0^j | \psi \rangle}{\langle \psi | \psi \rangle} =  \frac{1}{2} \delta_{a b} Q^a Q^b + \sum_{n=1}^\infty \sum_{a=1}^N \frac{\lVert J_n^a |\psi\rangle \rVert^2}{\langle \psi | \psi \rangle} \ge \frac{1}{2} \delta_{a b} Q^a Q^b \,,
\end{equation}
where the bound is saturated iff $| \psi \rangle$ is a current primary, $J_n^a | \psi \rangle = 0$ for all $n>0$.
\end{proof}

Note that for \emph{any} state $| \psi \rangle$ of charge $Q^a$, lowering using $J_n^a$, $n>0$ always produces a current primary with $h^j = \frac{1}{2} \delta_{a b} Q^a Q^b$. However, there need not be a charge $Q^a$ state saturating the bound $\hat{h} \ge 0$. For a unitary, modular-invariant CFT, \eqref{eqn:SpecFlow} implies that such an ``extremal'' state exists iff $Q^a$ lies on the period lattice $\Gamma$.

\section{Toroidal compactification of the bosonic string}
\label{app:toroidal}
In this appendix we derive the effective action for toroidal compactifications
of the bosonic string and use it to determine the long range forces between
string modes. This is similar to the calculation in~\cite{Heidenreich:2019zkl}, appendix B, where an analogous formula was derived for heterotic string theory on a torus.

Consider bosonic string theory compactified down to $D$ dimensions on a $k =
26 - D$ dimensional torus. The mass spectrum is
\begin{equation}
  \frac{\alpha'}{4} m^2 = \frac{1}{2} \delta^{a b} Q_a Q_b + N - 1 =
  \frac{1}{2} \delta^{\tilde{a}  \tilde{b}} \tilde{Q}_{\tilde{a}}
  \tilde{Q}_{\tilde{b}} + \tilde{N} - 1,
\end{equation}
where $N, \tilde{N}$ are non-negative integers and $Q_a$, $\tilde{Q}_{\alpha}$
are the charges under the left and right-moving conserved currents $J^a,
\tilde{J}^{\tilde{a}}$, with $a, \tilde{a} = 1, \ldots, k$. The possible left/right moving
charges $Q_A = (Q_a, \tilde{Q}_{\tilde{a}})$, $A = 1,
\ldots, 2 k$ live on a lattice $Q_A \in \Gamma$ that is
even self-dual with respect to the signature $(k, k)$ metric $\eta^{A B} = \Bigl(\begin{smallmatrix}\delta^{a b} & 0\\0 & - \delta^{\tilde{a}  \tilde{b}}\end{smallmatrix}\Bigr)$, i.e.,
\begin{equation}
  \forall \, Q \in \Gamma, \eta^{AB} Q_A Q_B \in 2\mathbb{Z}, \quad \text{and} \quad
 \Gamma^{\ast} \equiv \bigl\{ Q'_A \bigm| \forall\, Q_B \in \Gamma, \eta^{C D} Q_C Q'_D \in
  \mathbb{Z} \bigr\} =  \Gamma. \label{eqn:latticeCons}
\end{equation}
To write the
mass formula in the same notation, we define the positive-definite symmetric
bilinear form $\varphi^{A B} = \delta^{AB} = \Bigl(\begin{smallmatrix}
  \delta^{a b} & 0\\
  0 & \delta^{\tilde{a}  \tilde{b}}
\end{smallmatrix}\Bigr)$, so that
\begin{equation}
  \frac{\alpha'}{4} m^2 = \frac{1}{4}  (\varphi^{A B} + \eta^{A B}) Q_A Q_B +
  N - 1 = \frac{1}{4}  (\varphi^{A B} - \eta^{A B}) Q_A Q_B + \tilde{N} - 1 .
\end{equation}
Taking the sum and difference of these two equations gives:
\begin{align}
  m^2 &= \frac{1}{\alpha'} \varphi^{A B} Q_A Q_B + \frac{2}{\alpha'}  (N +
  \tilde{N} - 2), & \frac{1}{2} \eta^{A B} Q_A Q_B + N - \tilde{N} &= 0 .
  \label{eqn:massLevel}
\end{align}
Notice that the metric $\eta^{A B}$ appears in the charge lattice constraints
\eqref{eqn:latticeCons} and the level-matching condition, whereas the mass
spectrum depends on $\varphi^{A B}$. It is well known that any two even
self-dual lattices of the same indefinite signature are related by a Lorentz
transformation $Q_A \rightarrow Q_B \Lambda_{\;\; A}^B$ for
$\Lambda^A_{\;\; C} \eta^{C D} \Lambda^B_{\;\; D} =
\eta^{A B}$. Thus, after a suitable change of basis, the charge lattice can
take any chosen, canonical form $\Gamma_0$, but in general $\varphi^{A B}$ is
no longer equal to $\delta^{AB}$ in such a basis. Thus, there is a moduli space of
worldsheet CFTs that is isomorphic (up to discrete identifications) to the
space of positive-definite symmetric bilinear forms $\varphi^{A B}$ related to
$\delta^{AB}$ by a Lorentz transformation.

To describe this moduli space explicitly, note that in the original basis
$\varphi^{A B}$ satisfies\footnote{We also have $\varphi^A_{\;\; A}
= \varphi^{A B} \eta_{A B} = 0$, but this turns out to be a consequence of the
other conditions.}
\begin{equation}
  \varphi^{A B} = \varphi^{B A} \succ 0, \qquad \varphi^A_{\;\; B}
  \varphi^B_{\;\; C} = \delta^A_C, \label{eqn:phiCond}
\end{equation}
where ``$(\cdots) \succ 0$'' indicates a positive-definite matrix and indices
are raised and lowered with the metric $\eta_{A B}$, e.g.,
$\varphi^A_{\;\; C} \equiv \varphi^{A B} \eta_{B C}$.\footnote{Note
that $\varphi_{A B} \equiv \eta_{A C} \eta_{B D} \varphi^{C D}$ is also the
\emph{inverse} of $\varphi^{A B}$, due to the condition
$\varphi^A_{\;\; B} \varphi^B_{\;\; C} = \delta^A_C$.}
Since these conditions are $O (k, k)$ invariant, they hold in any basis.
Conversely, let $\varphi^{A B}$ be any matrix satisfying \eqref{eqn:phiCond}.
Since $\varphi^{A B}$ is symmetric and positive definite, we can choose a
basis where
\begin{equation}
  \varphi^{A B} = \delta^{A B},
\end{equation}
by the usual Gram-Schmidt procedure. When the change of basis does not
necessarily preserve the form of the metric $\eta_{A B}$, the constraint
$\varphi^A_{\;\; B} \varphi^B_{\;\; C} = \delta^A_C$
implies that $\sum_B \eta_{A B} \eta_{B C} = \delta_{A C}$ in this basis;
since $\eta_{A B}$ is symmetric, after an orthogonal transformation preserving
$\varphi^{A B} = \delta^{A B}$ we can fix
\begin{equation}
  \eta_{A B} = \diag (1, \ldots, 1, - 1, \ldots, - 1),
\end{equation}
as before. Thus the conditions \eqref{eqn:phiCond} are necessary and
sufficient to ensure that $\varphi^{I J}$ is indeed a Lorentz transformation
of $\delta^{I J}$.

Expanded about any chosen vacuum $\varphi_0^{A B}$, we have $\varphi^{A B} =
\varphi_0^{A B} + \lambda^{A B} + O (\lambda^2)$. Fixing $\varphi_0^{A B} =
\delta^{A B}$ by an $O (k, k)$ transformation, this becomes
\begin{equation}
  \varphi = \begin{pmatrix}
    1 & 0\\
    0 & 1
  \end{pmatrix} + \begin{pmatrix}
    0 & \lambda\\
    \lambda^{\top} & 0
  \end{pmatrix} + O (\lambda^2), \label{eqn:phiCanon}
\end{equation}
where the moduli $\lambda^{a \tilde{a}}$ are those generated by the $(1, 1)$
worldsheet primaries $\lambda^{a \tilde{a}} (z, \bar{z}) = J^a (z) 
\tilde{J}^{\tilde{a}} (\bar{z})$.

The low-energy effective action involves the gauge fields $F^I = (F^a,
F^{\tilde{a}})$ corresponding to the conserved currents $J^a,
\tilde{J}^{\tilde{a}}$, the moduli $\varphi^{A B}$ corresponding to the exactly
marginal operators $\lambda^{a \tilde{a}}$, as well as the spacetime metric
$g_{\mu \nu}$, dilaton $\Phi_D$ and Kalb-Ramond two-form $B_2$ as usual. One
finds\footnote{Here we ignore the tachyon as well as its KK modes. The latter
can become massless at special values of $R$, but generically have non-zero
$m^2$ (whether positive or negative).}
\begin{align}
  S &= \frac{1}{2 \kappa_D^2}\! \int \! \dd^D x \sqrt{- g}\, e^{- 2 \Phi_D}\! 
  \biggl[ R + 4 (\nabla \Phi_D)^2 \!-\! \frac{1}{2}  | \tilde{H}_3 |^2 \!-\!
  \frac{\alpha'}{2} \varphi_{A B} F^A \!\cdot\! F^B \!+\! \frac{1}{8} \nabla
  \varphi_{A B} \!\cdot\! \nabla \varphi^{A B} \biggr], \nonumber\\
  &\noeq \text{where} \quad \varphi^{A B} = \varphi^{B A} \succ 0, \qquad
  \varphi^A_{\;\; B} \varphi^B_{\;\; C} = \delta^A_C,
  \qquad \dd \tilde{H}_3 = - \frac{\alpha'}{2} \eta_{A B} F^A \wedge F^B . 
  \label{eqn:effAction}
\end{align}

This can be shown inductively, where the starting $D = 26, k = 0$ case is
familiar from textbooks. For the inductive step, we begin with the effective
action \eqref{eqn:effAction} in $D$ dimensions and compactify on a circle $y
\cong y + 2 \pi R$ with the ansatz:
\begin{align}
  \dd s^2_D &= \dd s_d^2 + e^{2 \sigma}  (\dd y + C_1)^2, &
  \Phi_D &= \Phi_d + \frac{1}{2} \sigma, \nonumber\\
  \tilde{H}_3^{(D)} &= \tilde{H}_3 + \tilde{H}_2 \wedge (\dd y + C_1),
  & F^A_{(D)} &= \tilde{F}^A + \dd \alpha^A \wedge (\dd y + C_1) ,
\end{align}
where $d = D - 1$. Then the action reduces to:
\begin{multline}
  S = \frac{1}{2 \kappa_d^2} \int \dd^d x \sqrt{- g} e^{- 2 \Phi_d} 
  \biggl[ R + 4 (\nabla \Phi_d)^2 - \frac{1}{2}  | \tilde{H}_3 |^2 -
  \frac{\alpha'}{2} \varphi_{A B}  \tilde{F}^A \cdot \tilde{F}^B - \frac{1}{2}
  e^{2 \sigma}  | G_2 |^2 \\ - \frac{1}{2} e^{- 2 \sigma} | \tilde{H}_2 |^2 
  + \frac{1}{8} \nabla \varphi_{A B} \cdot \nabla \varphi^{A B} - (\nabla
  \sigma)^2 - \frac{\alpha'}{2} e^{- 2 \sigma} \varphi_{A B} \nabla \alpha^A
  \cdot \nabla \alpha^B \biggr]\,,   \label{eqn:redAction}
\end{multline}
where $\kappa_d^2 = \kappa_D^2 / (2 \pi R)$, $G_2 = \dd C_1$, and
\begin{align}
  \dd \tilde{H}_3 &= - \frac{\alpha'}{2} \eta_{A B} 
  \tilde{F}^A \wedge \tilde{F}^B - \tilde{H}_2 \wedge G_2, & \dd
  \tilde{H}_2 &= - \alpha' \eta_{A B}  \tilde{F}^A \wedge \dd \alpha^B,
  & \dd \tilde{F}_2^A &= \dd \alpha^A \wedge G_2.
\end{align}
One can solve the modified Bianchi identities for $\tilde{H}_2$ and $\tilde{F}_2^A$ as
follows:
\begin{equation}
  \tilde{F}_2^A = F_2^A + \alpha^A G_2, \qquad \tilde{H}_2 = H_2 - \alpha'
  \eta_{A B} \alpha^A F_2^B - \frac{\alpha'}{2} \eta_{A B} \alpha^A \alpha^B
  G_2, \qquad \dd F_2^A = \dd H_2 = 0.
\end{equation}
This gives
  $\dd \tilde{H}_3 = - \frac{\alpha'}{2} \eta_{A B} F^A \wedge F^B - H_2
  \wedge G_2 = - \frac{\alpha'}{2} \eta_{\dot{A} \dot{B}} F^{\dot{A}} \wedge
  F^{\dot{B}}$,
where
\begin{equation}
  F^{\dot{A}} = \begin{pmatrix}
    F^A\\
    G_2\\
    H_2
  \end{pmatrix}, \qquad \eta_{\dot{A} \dot{B}} = \begin{pmatrix}
    \eta_{A B} & 0 & 0\\
    0 & 0 & \frac{1}{\alpha'}\\
    0 & \frac{1}{\alpha'} & 0
  \end{pmatrix},
\end{equation}
are the multiplet of gauge fields and the indefinite-signature metric in the
reduced theory. The inverse of this metric also shows up in the level matching
condition:
\begin{equation}
  0 = n w + \frac{1}{2} \eta^{A B} Q_A Q_B + N - \tilde{N} = \frac{1}{2}
  \eta^{\dot{A}  \dot{B}} Q_{\dot{A}} Q_{\dot{B}} + N - \tilde{N}, \qquad
  Q_{\dot{A}} = \begin{pmatrix}
    Q_A\\
    \frac{n}{R}\\
    \frac{w R}{\alpha'}
  \end{pmatrix},
\end{equation}
where $n$ and $w$ are the KK mode and winding numbers on the circle.

To finish putting the reduced action \eqref{eqn:redAction} into the form
\eqref{eqn:effAction}, we identify:
\begin{equation}
  \varphi_{\dot{A} \dot{B}} = \begin{pmatrix}
    \varphi_{A B} + \alpha' e^{- 2 \sigma} \alpha_A \alpha_B & \varphi_{A C}
    \alpha^C + \frac{\alpha'}{2} e^{- 2 \sigma} \alpha^2 \alpha_A & - e^{- 2
    \sigma} \alpha_A\\
    \varphi_{B C} \alpha^C + \frac{\alpha'}{2} e^{- 2 \sigma} \alpha^2
    \alpha_B & \varphi_{C D} \alpha^C \alpha^D + \frac{1}{\alpha'} e^{2
    \sigma} + \frac{\alpha'}{4} e^{- 2 \sigma} \alpha^4 & - \frac{1}{2} e^{- 2
    \sigma} \alpha^2\\
    - e^{- 2 \sigma} \alpha_B & - \frac{1}{2} e^{- 2 \sigma} \alpha^2 &
    \frac{1}{\alpha'} e^{- 2 \sigma}
  \end{pmatrix},
\end{equation}
by examining the gauge kinetic terms. One can check that this satisfies the
constraints \eqref{eqn:phiCond} and moreover that
\begin{equation}
  \frac{1}{8} \nabla \varphi_{\dot{A}  \dot{B}} \cdot \nabla \varphi^{\dot{A} 
  \dot{B}} = \frac{1}{8} \nabla \varphi_{A B} \cdot \nabla \varphi^{A B} -
  (\nabla \sigma)^2 - \frac{\alpha'}{2} e^{- 2 \sigma} \varphi_{A B} \nabla
  \alpha^A \cdot \nabla \alpha^B,
\end{equation}
so the reduced action is again of the form \eqref{eqn:effAction}, proving that
the effective action takes this form in general by induction on the number of
compact dimensions $k$.

To compute the long range forces mediated by the massless fields appearing in
the effective action \eqref{eqn:effAction}, we choose a basis where the
background values of the moduli take a canonical form $\varphi^{A B}_0 =
\delta^{A B}$ so that $\varphi^{A B} = \delta^{A B} + \lambda^{A B} + O
(\lambda^2)$ as in \eqref{eqn:phiCanon}, and the kinetic terms for the $J
\tilde{J}$ moduli become
\begin{equation}
  S_{\text{mod}} = \frac{1}{16 \kappa_D^2}\! \int \! \dd^D x \sqrt{- g} e^{- 2
  \Phi_D} \nabla \varphi^{A B} \cdot \nabla \varphi_{A B} = - \frac{1}{8
  \kappa_D^2} \! \int \! \dd^D x \sqrt{- g} e^{- 2 \Phi_D} \delta_{a b}
  \delta_{\tilde{a}  \tilde{b}} \nabla \lambda^{a \tilde{a}} \cdot \nabla
  \lambda^{b \tilde{b}} .
\end{equation}
Using \eqref{eqn:massLevel}, we find
\begin{equation}
  \frac{\partial m^2}{\partial \lambda^{a \tilde{a}}} = \frac{2}{\alpha'} Q_a 
  \tilde{Q}_{\tilde{a}} \qquad \Rightarrow \qquad \frac{\partial m}{\partial
  \lambda^{a \tilde{a}}} = \frac{Q_a  \tilde{Q}_{\tilde{a}}}{\alpha' m} .
\end{equation}
Thus, the $J \tilde{J}$-mediated long-range force coefficient (defined in~\eqref{eqn:selfforce}) between
two particles of mass $m, m'$ and charge $Q_A, Q'_A$ is
\begin{equation}
  \mathcal{F}^{\lambda} = - \frac{4 \kappa_D^2}{{\alpha'}^2}  \frac{\delta^{a
  b} \delta^{\tilde{a} \tilde{b}} Q_a Q_b'  \tilde{Q}_{\tilde{a}}
  \tilde{Q}_{\tilde{b}}'}{m m'} . \label{eqn:lambdaforce}
\end{equation}
The gauge force is straightforward to read off from the effective action:
\begin{equation}
  \mathcal{F}^{\text{gauge}} = \frac{2 \kappa_D^2}{\alpha'} \varphi^{A B} Q_A
  Q_B' = \frac{2 \kappa_D^2}{\alpha'}  (\delta^{a b} Q_a Q_b' +
  \delta^{\tilde{a} \tilde{b}}  \tilde{Q}_{\tilde{a}} \tilde{Q}_{\tilde{b}}'), \label{eqn:gaugeforce}
\end{equation}
as are the graviton and dilaton contributions:\footnote{See \cite{Heidenreich:2019zkl} for a detailed treatment of the dilaton contribution.}
\begin{equation}
  \mathcal{F}^{\text{grav} + \Phi} = - \frac{D - 3}{D - 2} \kappa_D^2 m m' -
  \frac{1}{D - 2} \kappa_D^2 m m' = - \kappa_D^2 m m' . \label{eqn:gravdilaton}
\end{equation}
In net, we obtain:
\begin{equation}
  \mathcal{F}= \frac{2 \kappa_D^2}{\alpha'}  (\delta^{a b} Q_a Q_b' +
  \delta^{\tilde{a} \tilde{b}}  \tilde{Q}_{\tilde{a}} \tilde{Q}_{\tilde{b}}')
  - \frac{4 \kappa_D^2}{{\alpha'}^2}  \frac{\delta^{a b} \delta^{\tilde{a}
  \tilde{b}} Q_a Q_b'  \tilde{Q}_{\tilde{a}} \tilde{Q}_{\tilde{b}}'}{m m'} -
  \kappa_D^2 m m' .
\end{equation}
As observered in~\cite{Heidenreich:2019zkl}, this conveniently factors:
\begin{equation}
  \mathcal{F}= - \frac{4 \kappa_D^2}{{\alpha'}^2 m m'}  \biggl(
  \frac{\alpha'}{2} m m' - \delta^{a b} Q_a Q_b' \biggr) \biggl(
  \frac{\alpha'}{2} m m' - \delta^{\tilde{a} \tilde{b}} \tilde{Q}_{\tilde{a}}
  \tilde{Q}_{\tilde{b}}' \biggr) .
\end{equation}

\bibliographystyle{JHEP}
\bibliography{refs}

\end{document}